\documentclass[10pt, conference,anonymous]{IEEEtran}

\usepackage{subcaption}

\usepackage{cite}
\usepackage{amsmath,amssymb,amsfonts,amsthm,comment}
\usepackage{graphicx}
\usepackage{textcomp}
\usepackage{xcolor}
\usepackage{tikz}
\usepackage[ruled, noline]{algorithm2e}
\usepackage{flushend}
\usepackage{pgfplots,tikz}
\usetikzlibrary{patterns}
\usepgfplotslibrary{groupplots}
\pgfplotsset{width=7cm,compat=1.8}

\usepackage{enumitem}
\usepackage{multirow}
\usepackage[T1]{fontenc}
\usepackage{booktabs}
\usepackage{hyperref}

\usepackage{xcolor}
\usepackage{csquotes}
\usepackage{dashbox}
\usepackage{array}
\usepackage{setspace}

\usetikzlibrary{shapes.callouts}
\usetikzlibrary{patterns, backgrounds}
\pgfplotsset{ytick style={draw=none}, xtick style={draw=none}}

\usepackage{listings}

\usepackage{trace}

\usepackage{makeidx}
\usepackage{mdframed}

\usepackage{fancyhdr}

\usepackage{xspace}
\usepackage{xparse}
\usepackage{booktabs}
\usepackage{mathtools}
\usepackage{amsthm}
\usepackage{scalerel}
\usepackage{bm}
\DeclareMathAlphabet{\mathcal}{OMS}{cmsy}{m}{n}
\SetMathAlphabet{\mathcal}{bold}{OMS}{cmsy}{b}{n}
\newcommand{\bnm}{\begin{newmath}}
\newcommand{\enm}{\end{newmath}}

\newcommand{\bea}{\begin{neweqnarrays}}%
\newcommand{\eea}{\end{neweqnarrays}}%

\newcommand{\bne}{\begin{newequation}}
\newcommand{\ene}{\end{newequation}}

\newcommand{\bal}{\begin{newalign}}
\newcommand{\eal}{\end{newalign}}

\newenvironment{newalign}{\begin{align*}%
\setlength{\abovedisplayskip}{4pt}%
\setlength{\belowdisplayskip}{4pt}%
\setlength{\abovedisplayshortskip}{6pt}%
\setlength{\belowdisplayshortskip}{6pt} }{\end{align*}}

\newenvironment{newmath}{\begin{displaymath}%
\setlength{\abovedisplayskip}{4pt}%
\setlength{\belowdisplayskip}{4pt}%
\setlength{\abovedisplayshortskip}{6pt}%
\setlength{\belowdisplayshortskip}{6pt} }{\end{displaymath}}

\newenvironment{neweqnarrays}{\begin{eqnarray*}%
\setlength{\abovedisplayskip}{-1pt}%
\setlength{\belowdisplayskip}{-1pt}%
\setlength{\abovedisplayshortskip}{1pt}%
\setlength{\belowdisplayshortskip}{1pt}%
\setlength{\jot}{-0.4in} }{\end{eqnarray*}}

\newenvironment{newequation}{\begin{equation}%
\setlength{\abovedisplayskip}{4pt}%
\setlength{\belowdisplayskip}{4pt}%
\setlength{\abovedisplayshortskip}{6pt}%
\setlength{\belowdisplayshortskip}{6pt} }{\end{equation}}
\newcommand{\paraheading}[1]{\noindent\textbf{#1\;}}

\newcounter{ctr}

\newenvironment{newitemize}{%
\begin{list}{\mbox{}\hspace{5pt}$\bullet$\hfill}{\labelwidth=15pt%
\labelsep=4pt \leftmargin=12pt \topsep=3pt%
\setlength{\listparindent}{\saveparindent}%
\setlength{\parsep}{\saveparskip}%
\setlength{\itemsep}{3pt} }}{\end{list}}

\newlength{\saveparindent}
\newlength{\saveparskip}
\DeclareMathSymbol{\mlq}{\mathord}{operators}{``}
\DeclareMathSymbol{\mrq}{\mathord}{operators}{`'}

\def\suchthat{\: |\:}

\newcommand{\E}{{\rm I\kern-.3em E}}
\newcommand{\Prob}[1]{{\Pr\left[\,{#1}\,\right]}}

\newcommand{\xor}{\oplus}

\newcommand{\send}{{\ensuremath{\rightarrow}}}

\newcommand{\getsr}{{\:{\leftarrow{\hspace*{-3pt}\raisebox{.75pt}{$\scriptscriptstyle\$$}}}\:}}

\providecommand{\true}{\mathsf{true}}
\providecommand{\false}{\mathsf{false}}

\def\negl{\mathsf{negl}}

\newcommand{\secref}[1]{\mbox{Section~\ref{#1}}}
\newcommand{\apref}[1]{\mbox{Appendix~\ref{#1}}}
\newcommand{\appref}[1]{\apref{#1}}

\newcommand{\thref}[1]{\mbox{Thm.~\ref{#1}}}

\newcommand{\propref}[1]{\mbox{Proposition~\ref{#1}}}

\newcommand{\figref}[1]{\mbox{Fig.~\ref{#1}}}

\renewcommand{\eqref}[1]{\mbox{(\ref{#1})}}

\newcommand{\tabfontsize}{\scriptsize}
\newcommand{\gamesfontsize}{\footnotesize}

\providecommand{\setstretch}[1]{}
\newcommand{\mpage}[2]{\begin{minipage}[t]{#1\textwidth}\setstretch{1.03}\gamesfontsize\vspace{0pt}  #2 \end{minipage}}
\newcommand{\framedminipage}[2]{\framebox{\mpage{#1}{#2}}}

\newcommand{\hfpagess}[4]{
  \begin{tabular}[t]{c@{\hspace*{.2em}}c}
    \framedminipage{#1}{#3}
    & \framedminipage{#2}{#4}
  \end{tabular}
}
\newcommand{\hfpagesss}[6]{
  \begin{tabular}[t]{c@{\hspace*{.5em}}c@{\hspace*{.5em}}c}
    \framedminipage{#1}{#4}
    & \framedminipage{#2}{#5}
    & \framedminipage{#3}{#6}
  \end{tabular}
}

\newcommand{\calA}{{\mathcal A}}
\newcommand{\calC}{{\mathcal{C}}}

\newcommand{\calX}{{\mathcal{X}}}
\newcommand{\calY}{{\mathcal{Y}}}

\newcommand{\alphabet}{\Sigma}

\def \key {k}
\def \keylen {\kappa}

\def \ctxt {ct}

\providecommand{\concat}{\,\|\,}

\def \part {part}

\def \adv {{\mathcal A}}

\newcommand{\vecx}{\vec{x}}

\newcommand{\advA}{{\mathcal{A}}}

\newcommand{\advB}{{\mathcal{B}}} %

\newcommand{\advG}{{\mathcal{G}}}

\newcommand{\G}{\advG}

\newcommand{\Colon}{{\::\;}}

\renewcommand{\paragraph}[1]{\vspace*{6pt}\noindent\textbf{#1}\;}

\newcounter{mytable}
\def\mytable{\begin{centering}\refstepcounter{mytable}}
\def\endmytable{\end{centering}}

\newcounter{myfig}
\def\myfig{\begin{centering}\refstepcounter{myfig}}
\def\endmyfig{\end{centering}}

\newtheorem{theorem}{Theorem}[section]

\newtheorem{proposition}{Proposition}

\def \blackslug{\hbox{\hskip 1pt \vrule width 4pt height 8pt
    depth 1.5pt \hskip 1pt}}
\def \qed{\quad\blackslug\lower 8.5pt\null\par}

\def \zo  {\{0,1\}}

\newcommand{\ts}{t}

\newcounter{mynote}[section]

\newcommand{\change}[1]{{\color{blue} #1}}

\newcommand\ignore[1]{}

\newcounter{rcnote}[section]

\newcommand{\mytab}{\hspace*{.4cm}}

\newcommand{\rhf}[2]{R_{f, \gamma}}

\def\trans{\textnormal{\ensuremath{T}}}

\def\H{\ensuremath{\mathsf{H}}\xspace}
\def\hmac{\ensuremath{\mathsf{HMAC}}\xspace}

\newcommand{\return}{\textrm{Return}~}

\DeclareDocumentCommand{\edist}{o o}{
  \ensuremath{
    \IfNoValueTF{#1}{{d}}{{\sf d}(#1,#2)}
  }
}

\newcommand{\skegen}{\ensuremath{\mathsf{K}}}
\newcommand{\skeenc}{\ensuremath{\mathsf{E}}}
\newcommand{\skedec}{\ensuremath{\mathsf{D}}}

\newcommand{\ske}{\mathcal{E}}

\newcommand{\states}{\mathcal{S}}

\DeclareMathSymbol{\mlq}{\mathord}{operators}{``}
\DeclareMathSymbol{\mrq}{\mathord}{operators}{`'}

\newcommand{\name}{\textsf{eTAP}\xspace}
\newcommand{\stap}{\name}
\newcommand{\dt}{x}
\newcommand{\payload}{v}
\newcommand{\dtp}{\payload}
\newcommand{\da}{y}
\newcommand{\ei}{e}
\newcommand{\eis}{\ei_s}
\newcommand{\eir}{\ei_r}

\newcommand{\LT}{X}
\newcommand{\LA}{Y}

\newcommand{\Gb}{\textsf{Gb}}
\newcommand{\Ev}{\textsf{Ev}}
\newcommand{\En}{\textsf{En}}
\newcommand{\De}{\textsf{De}}
\newcommand{\fone}{f_1}
\newcommand{\ftwo}{f_2}
\newcommand{\const}{c}
\newcommand{\cone}{\const_1}
\newcommand{\ctwo}{\const_2}

\newcommand{\TS}{\text{TS}\xspace}
\newcommand{\AS}{\text{AS}\xspace}
\newcommand{\TAP}{\text{TAP}\xspace}
\newcommand{\TC}{\text{TC}\xspace}

\newcommand{\wlabel}{L}
\newcommand{\encd}{\tilde{d}}
\newcommand{\encs}{\tilde{s}}
\newcommand{\ench}{\tilde{h}}

\newcommand{\dfa}{\Gamma}
\renewcommand{\trans}{\delta}
\newcommand{\transstring}{{\vec{\delta}}}
\newcommand{\finalstates}{\mathcal{S}_F}

\newcommand{\nstates}{q}
\newcommand{\Trans}{\Delta}
\newcommand{\Transstring}{{\vec{\Delta}}}
\newcommand{\State}{S}
\newcommand{\state}{s}

\newcommand{\delay}{\tau}
\newcommand{\sects}{\bm{\key_{\scaleto{T}{3pt}}}}
\newcommand{\secas}{\bm{\key_{\scaleto{A}{3pt}}}}
\newcommand{\sectservice}{\bm{\key_{\scaleto{\TS}{3pt}}}}
\newcommand{\secaservice}{\bm{\key_{\scaleto{\AS}{3pt}}}}

\newcommand{\lsb}{\textsf{lsb}}
\newcommand{\proto}{{\scaleto{\textsf{etap}}{5pt}}}
\newcommand{\obliv}{\mathsf{Obliv}_\calA^{\proto}}
\newcommand{\auth}{\mathsf{Auth}_\calA^{\proto}}
\newcommand{\privone}{\mathsf{Priv}_\advB^{\proto, 1}}

\newcommand{\privts}{\mathsf{Priv}_{\rm TS}}
\newcommand{\privas}{\mathsf{Priv}_{\rm AS}}
\newcommand{\Ytilde}{Y'}

\newcommand{\payloadk}{k_v}

\newcommand{\tsexec}{{\sf TSExec}}
\newcommand{\asexec}{{\sf ASExec}}
\newcommand{\tapexec}{{\sf TAPExec}}
\newcommand{\cktgarbling}{{\sf CktGarbling}}

\renewcommand{\change}[1]{{\color{black} #1}}
\def\BibTeX{{\rm B\kern-.05em{\sc i\kern-.025em b}\kern-.08em
    T\kern-.1667em\lower.7ex\hbox{E}\kern-.125emX}}

\usepackage{url}

\usepackage{caption}
\begin{document}

\def\papertitle{Data Privacy in Trigger-Action Systems}

\title{\papertitle}

\author{
\IEEEauthorblockN{Yunang Chen, \ 
  Amrita Roy Chowdhury, \ 
  Ruizhe Wang, \\[2pt]
  Andrei Sabelfeld\IEEEauthorrefmark{2}, \ 
  Rahul Chatterjee, \ 
  Earlence Fernandes
}
\vspace{5pt}
\IEEEauthorblockA{University of Wisconsin--Madison \hspace{1in} \IEEEauthorrefmark{2} Chalmers University of Technology}
}

\maketitle

\pagestyle{plain}

\begin{abstract}
  Trigger-action platforms (TAPs) allow users to connect independent 
  web-based or IoT services to achieve useful automation. They provide a simple
  interface that helps end-users create trigger-compute-action rules that pass
  data between disparate Internet services.  Unfortunately, TAPs
  introduce a large-scale security risk: if they are compromised, attackers will
  gain access to sensitive data for millions of users. To avoid this risk, we propose \name, a privacy-enhancing trigger-action platform
  that executes trigger-compute-action rules
  without accessing users' private data in plaintext or learning anything about
  the results of the computation. We use garbled
  circuits as a primitive, and leverage the unique structure of
  trigger-compute-action rules to make them practical.  We formally state and
  prove the security guarantees of our protocols. We prototyped \name, which supports the most commonly used operations on
  popular commercial TAPs like IFTTT and Zapier.  Specifically, it supports
  Boolean, arithmetic, and string operations on private trigger data and can run
  100\% of the top-500 rules of IFTTT users and 93.4\% of all publicly-available rules on
  Zapier. Based on ten existing rules that exercise a wide variety of operations, we show that \name has a modest performance impact: on average rule execution latency increases by 70~ms (55\%) and throughput reduces by 59\%.

\end{abstract}

\section{Introduction}
\label{sec:intro}


Trigger-action platforms (TAPs), such as IFTTT~\cite{IFTTT},
Zapier~\cite{zapier}, and Microsoft Power Automate~\cite{msft-flow} 
are web-based systems that enable users to stitch together their cyber-physical and
digital resources (e.g., IoT devices, GMail, Instagram, Slack) to achieve useful
automation.  TAPs provide a simple \emph{trigger-compute-action} paradigm and an
easy-to-use interface to program automation rules.

For example, using their smartphone, a user can setup a rule that
checks if an email contains the word ``confidential'' and, if so, sends an
SMS with the subject line and the sender's address to a pre-specified number
(\figref{fig:tas}). Instead of an SMS, the rule could also blink a smart light
whenever a matching email arrives. To execute this rule on a TAP, when an email
arrives (\emph{trigger}), the mail service (\emph{trigger service}) sends the
email to the TAP that runs the string search (\emph{computation}), which then
contacts an SMS gateway or a smart bulb service (\emph{action service})
with required information to perform the \emph{action}. 
We refer to the combination of trigger/action services and the TAP
as a trigger-action system --- a key ingredient for fulfilling the promise of the IoT~\cite{IFTTT:usage:2018}.  They
provide a layer of abstraction that enables trigger and action services to
develop APIs independently without worrying about compatibility with each
other.

These benefits unfortunately come at the high price of private data disclosure to the TAPs. Even the simple rule discussed above reveals the user's private emails to the TAP. As the TAP is the center of communication between triggers and actions, it can launch \emph{person-in-the-middle} attacks by invisibly collecting private information on all of its users, 
similar to what has already been happening on centralized ride-hailing platforms~\cite{databreach,theguardian}.
Due to the highly compatible nature of TAPs, 
this data includes location, voice commands, fitness data, pictures, files, etc.~\cite{IFTTT:Privacy} 
and is limited only by the variety of online services of users (e.g., IFTTT supports 600 services~\cite{IFTTT:use}). Commercial TAPs do not provide any technical protections for user data. For example, IFTTT's terms of use explicitly state that they collect personal data from third parties, and may pass it to other third parties, partners, or any company that might acquire IFTTT~\cite{IFTTT:Privacy}.

\begin{figure}[t!]
  \centering \includegraphics[width=0.48\textwidth]{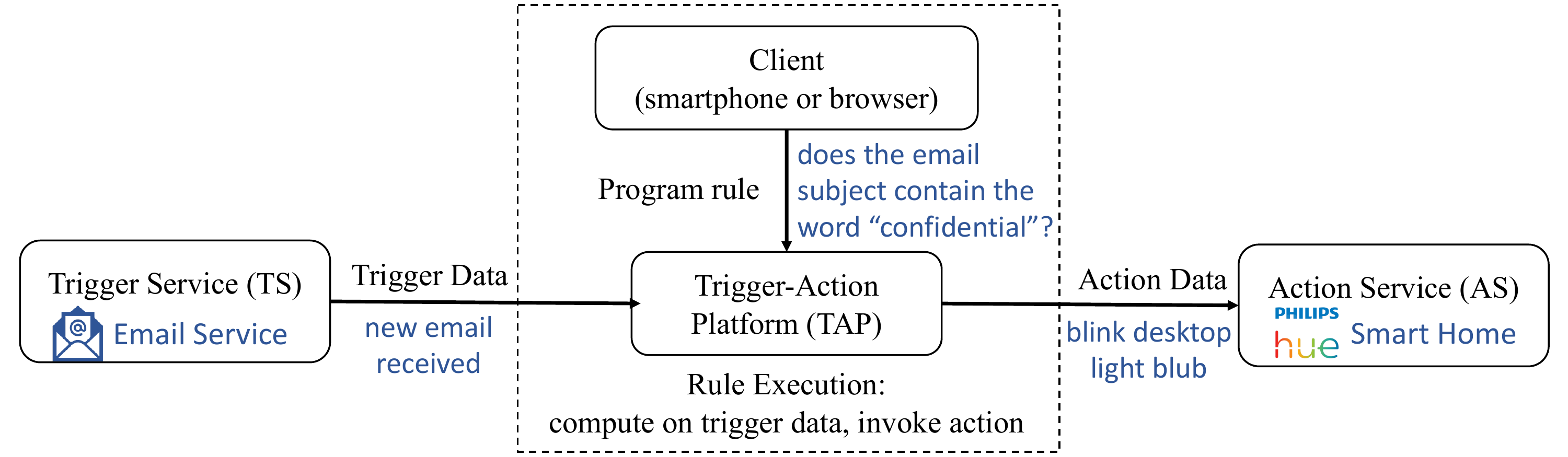}
    \caption{Overview of current trigger-action systems. The dataflow for the example rule  is illustrated in blue color: ``IF I receive an email containing the word `confidential', THEN blink my desktop smart light.''}
\label{fig:tas}
\end{figure}


Furthermore, because TAPs are widely-used centralized web services (e.g., IFTTT
has more than 20 million users~\cite{IFTTT:usersize}), they are attractive
targets for attackers. Breaches of cloud services are commonplace~\cite{equifax-breach, voter-breach, dropbox-breach,ebay-breach}. 
Attackers
sometimes even have continued access to the compromised service for days, and
even weeks before getting detected~\cite{zdnet:stack-overflow-hack,
  zdnet:telecom-hack, zdnet:exim-vuln}. 
A similar breach will have disastrous consequences for TAP users.  Such privacy
risks might discourage users as well as trigger/action services from using
TAPs. Indeed GMail, due to security and privacy concerns, pulled back some of
its APIs from IFTTT~\cite{gmail-ifttt-removal}.


In this paper, we introduce \name, an encrypted
trigger-action platform that executes user rules without accessing the
underlying user data in plaintext. Thus, \name provides confidentiality even when the
attacker fully controls the TAP. Although this problem fits in the general
framework of secure function evaluation
(SFE)~\cite{yao1982protocols,yao1986generate}, building a functional and secure
trigger-action platform with good performance requires overcoming several
challenges.

First, we desire confidentiality of user's data and authenticity of
computation when the TAP is compromised and acts
maliciously. While there are protocols for SFE that provide
security even if some parties act maliciously~\cite{lindell2008implementing,frederiksen2014faster}, these constructions are not yet  practical~\cite{lindell2008implementing,rindal2016faster}. 
Second, using off-the-shelf
protocols for SFE will require invasive changes to the architecture of trigger-action
systems that break the independence between trigger and action services, making them less useful. Third, running arbitrary computations on the TAP using SFE will be inefficient. 

We leverage the unique structure and threat model of trigger-action systems to
overcome these challenges. At a high-level, we create a trusted generator of
garbled circuits (GCs). 
This allows \name to use semi-honest implementations of SFE coupled with a few
efficient extensions, which we contribute with security proofs, to achieve
security against a fully malicious circuit evaluator.

In our setting, the user's smartphone, a standard component in TAP design, plays
the role of a trusted circuit generator that periodically generates and
transmits garbled circuits to the untrusted TAP. The trigger service garbles
sensitive data when it is available and calls the TAP, which then executes the
circuit and contacts the action service with the (garbled) results. The action service
performs security checks and then executes the action.
We assume that the user's phone is fully trusted, while TAP is malicious. An
attacker interested in compromising a large number of users is more likely
to try compromising the TAP than the user's phone.  To maintain the same level
of trust as current TAPs provide, we treat the trigger and action services as
semi-honest --- they follow the protocol but can be inquisitive --- and they should
not learn 
any new private information that they do not learn in the current setting.

To overcome the challenge concerning the efficiency of arbitrary computations,
we perform an analysis of the types of computations in popular commercial
trigger-action platforms. We show that the computations supported by TAPs are
stateless and use Boolean, arithmetic, or string operations. Most GC libraries support Boolean and arithmetic operations natively, but
none support string operations out of the box. \change{Existing work contributes oblivious deterministic finite automata that can match regular expressions~\cite{mohassel2010}. However, it does not support substring extraction and replacements --- a common operation in trigger-action systems.} We therefore introduce a novel
approach to efficiently encode a subset of fixed-length string operations
as 
Boolean circuits. We then use the standard GC approach to evaluate them securely
on the TAP. 
\change{Our approach also has the advantage of unifying all the formal security properties of \name rather than having a separate set of proofs for string operations}. 
\name can compute 93.4\% of all
rules published on Zapier that require computation and 100\% of the 500 most-used rules on IFTTT. \change{(Of course, \name supports all rules that do not require
any computation.)}

We formally prove the security of \name in the presence of a malicious
TAP (\secref{sec:sec_analysis}). We show that the malicious
TAP can execute user rules without learning the private data or tampering with the result of computation. \name also provides mutual secrecy between the
trigger and action services.

\name is a clean-slate approach to building trigger-action systems and lays a foundation for securing the data they handle. However, it does require some changes to current systems. First, the trigger/action services need to understand our protocols. We provide simple shims that they can use to upgrade their functionality while maintaining their independence and RESTful nature. Second, the user's client device takes on a more prominent role because it generates garbled circuits. As efficient circuits cannot be reused in general, the client has to periodically generate and transmit these circuits to the TAP. We estimate that this process has a modest impact:  
the trusted client is expected to transfer 61.7~MB of data per day for an average user.  
This is equivalent to the data consumed by uploading a one-minute of Full-HD video.

The paper offers the following contributions:
\begin{newitemize}
\item We design \name, the first trigger-action platform that can execute trigger-compute-action rules (Boolean, arithmetic, fixed-length string) without
  accessing the underlying trigger data in plaintext. 

\item We outline ideal security expectations of a privacy-sensitive
  trigger-action system, and formally prove that \name meets those security
  properties.

\item We implement and evaluate \stap. It can support 93.4\% of
  function-dependent rules used in Zapier and 100\% of the 500 most-used rules
  in IFTTT. We show that most functions can be evaluated in the order of
  milliseconds with about $2$x computational cost. Code is available at \url{https://github.com/EarlMadSec/etap}.
\end{newitemize}

\section{Background}\label{sec:background}



We discuss background information on trigger-action systems and the cryptographic primitives that we use.

\subsection{Trigger-Action Systems} \label{sec:background_tas}
Trigger-action systems allow stitching together disparate online services using
a trigger-compute-action paradigm to automate different tasks. There are three
main components of the system: trigger services (TSs), action services (ASs),
and a trigger-action platform (TAP). We also explicitly mention another
computing component: the user's client device that they use to interface with
the trigger-action system. \figref{fig:tas} shows the interactions between
different components.

Trigger and action services are online services for IoT or web apps.  There are a
plethora of such services such as Instagram, Slack, GMail, Amazon
Alexa, Samsung SmartThings, and many others. These services rely on REST APIs to
send and receive data, and each service may support several APIs to provide different functionalities.
They typically support the OAuth protocol~\cite{oauth},
which is used to delegate authorization. With OAuth tokens, a third party,
such as a TAP, can access APIs and execute trigger-compute-action rules.


Commercial TAPs are compatible
with hundreds of trigger and action services, allowing each trigger or action
service to focus on building their own REST APIs without worrying about
compatibility with each other. \change{Third-parties own a large majority of these services that integrate with IFTTT (e.g., LG, Samsung, Google).\footnote{\change{As of Aug 2020, 417 out of 522 services on IFTTT are third-party that require a user to login and authorize access to IFTTT.}}}


Additionally, modern TAPs also allow
performing non-trivial computation over the trigger data. The ability to modify
the trigger data provides great flexibility for TAPs to achieve compatibility
between trigger and action services (e.g., 
two calendar apps that use different date formats).
The TAP also uses operations
to decide whether or not it should send a message to the action
service (e.g., does the email contain the word ``confidential''). 
TAPs serve as a computation and communication hub.
Zapier has explicitly supported computation on trigger
data from the very beginning~\cite{zapier:filter, zapier:formatter}. IFTTT has recently
started to expose its computing interface to
end-users~\cite{IFTTT:applet}. Thus, trigger-action systems are evolving to be
\emph{trigger-compute-action} systems. We use these two terms interchangeably
throughout the paper.

Users interface with trigger-action system through a client device, typically a
smartphone. The user programs rules by selecting a trigger service, then
specifying a computation on that data using a library of functions,
and finally selecting an action to be run on the action service. As noted
before, the user also authorizes the \TAP to access their online services using
the client device.

\paragraph{Privacy and authenticity risks in current TAPs.}
Commercial TAPs operate on sensitive trigger data of millions of users, making
them an attractive target for attackers. If the TAP is compromised, the attacker
gains the privilege of the TAP --- unfettered access to user data and
resources. The types of data are limited only by the set of rules that users
create and the end-point services that the TAP supports. Commercial systems like
IFTTT support approximately 600 services currently~\cite{IFTTT:use}. The
sensitive information from these services can be emails (our earlier example),
data files, health information, voice commands, images, etc.

Fernandes et al.~\cite{fernandes2018decentralized} first noted this problem with TAPs, 
and discussed a more appropriate threat model where TAP can act maliciously. Under
this model, they addressed a sub-problem: preventing a compromised TAP from
misusing overprivileged OAuth tokens. Their work adds integrity to the rules,
but it does not allow any computation over the trigger data. 

By contrast, we target modern TAPs that allow computation over the trigger
data. Beyond integrity, we also aim to protect the \emph{privacy} of that
data. Our work provides a way for TAPs to compute on sensitive data without
seeing the plaintext, despite arbitrarily deviating from the protocol. We
believe such privacy risks might be preventing trigger-action systems from
achieving their true potential. Furthermore, we provide computational integrity
as well, thus subsuming prior work~\cite{fernandes2018decentralized}. 
\subsection{Cryptographic Primitives} \label{sec:crypto_tap}


\paragraph{Symmetric-key encryption scheme.} 
Let $\ske = (\skegen, \skeenc, \skedec)$ be a semantically secure encryption
scheme. The key generation function $\skegen(1^\keylen)$ generates a
$\keylen$-bit uniformly random key $\key$; the randomized encryption scheme
$\skeenc$ takes a message $x\in\calX$ and the generated key $\key$ as input and
outputs a cipher text $\ctxt \getsr\skeenc(\key, x)$; and the deterministic
decryption function takes a cipher text and the key $\key$ as input and outputs a
message, $x \gets \skedec(\key, \ctxt)$, or $\bot$ (if decryption fails).

We use an authenticated encryption scheme~\cite{bellare2000authenticated} that
achieves the IND-CCA security guarantee. This ensures both the
privacy and authenticity of plaintext.



\paragraph{Garbled circuits (GCs).} This is a cryptographic technique for secure
function evaluation (SFE)~\cite{yao2005design,bellare2000authenticated}.
Following Bellare et al.'s~\cite{bellare2012foundations} notations, a garbling
scheme $\G$ is a tuple of four functions $\G=(\Gb, \En, \De, \Ev)$. Let
$f\Colon\zo^n\to \zo^m$ denote the function to be evaluated securely. Here,
$\Gb$ is a randomized \emph{garbling \ function} that converts the function $f$
(represented as a Boolean circuit) into a \emph{garbled} circuit~$F$. It also
outputs encoding and decoding information $e$ and $d$ needed for encoding inputs
and decoding the outputs. As such, $(F, e, d)\getsr\Gb(1^\keylen, f)$, where
$\keylen$ is the security parameter. 
The \emph{encoding \ function} $(\En)$ encodes an input $x\in\zo^n$ using the
encoding information $e$, which is the set of labels corresponding to the value of each bit in $x$;  $X\gets\En(e, x)$. 
The \emph{evaluation \ function} $(\Ev)$ enables evaluation of the
garbled circuit $F$ over the garbled input $X$ to generate the garbled output
$Y\gets\Ev(F, X)$, which is the set of labels corresponding to the output wires.
Finally, the \emph{decoding \ function} $(\De)$ decodes the output of the
evaluation $y\gets\De(d, Y)$.

Garbling involves generating two random labels $\wlabel_1^w$ and
$\wlabel_0^w$ for each of its wires, representing the $\true$ and $\false$ value
for the wire $w$. 
A number of optimizations have been proposed to reduce the size of a garbled
circuit. One of them 
is the free XOR technique~\cite{kolesnikov2008improved}, which requires all wire labels to follow the form
$\wlabel_1^w = \wlabel_0^w \oplus \eir $, where $\eir$ is a string randomly
chosen by $\Gb$. This allows XOR gates in the circuit to be computed with only
the input wire labels.

Typically, GCs are used for 2-party secure function computations where two
parties with their respective private inputs $x_1$ and $x_2$ run the protocol
such that, no party learns more than $f(x_1,x_2)$ for a public function $f$. The
protocol works as follows. First, one of the parties, called the
\emph{generator}, uses the garbling function to
generate 
$(F, e, d)\getsr\Gb(1^\keylen, f)$. Next, it encodes its input as
$X_1\gets \En(x_1,e)$. The other party, called the \emph{evaluator}, receives $F$
and $X_1$ and also retrieves $X_2\gets\En(e, x_2)$ --- encoding of its private
input $x_2$ --- using an oblivious transfer (OT)~\cite{Rabin:OT} protocol with
the generator.  Following this, the evaluator runs the garbled circuit to
obtain $Y\gets \Ev(F,(X_1,X_2))$. Finally, either party can decode $Y$ to obtain
the final output $y\gets\De(d,Y)$.

A secure garbling scheme provides the following security properties~\cite{bellare2012foundations}: 
(a) \emph{Message \ obliviousness.} Given $(F, X)$, an adversary learns nothing
  about $x$ or $y$ (beyond what is known from $f$).
(b) \emph{Input \ privacy.} Given $(F, X, d)$, an adversary learns nothing
  about $x$ beyond what is known from $y$ and $f$. 
(c) \emph{Execution \ authenticity.} 
  Given a garbled input $X$, it is hard to find $\Ytilde$ such that
  $\Ytilde\ne\Ev(F, X)$ and $\De(d, \Ytilde)\ne\bot$.

We use these cryptographic primitives to design
\stap. In~\secref{sec:supported-functions}, we analyze existing TAPs to
understand what functions \stap must support. We give the detailed protocol
in~\secref{sec:design}, with its security proven in~\appref{sec:sec_analysis}.

\section{Analysis of Current Trigger-Action Systems}
\label{sec:supported-functions}
We analyze two popular commercial TAPs, IFTTT~\cite{IFTTT} and
Zapier~\cite{zapier} with the following goals in mind: (1) understand the
sensitive data that TAPs compute on; (2) establish that although TAPs offer a
variety of operations on data, they are not arbitrary and will fit well in a
garbled circuit framework; and (3) derive an abstract TAP computational model
that will help ensure our system supports realistic functionality.

\paragraph{Types of sensitive information.} The current trigger-action system design gives the cloud-based TAP complete access to trigger data. 
To better characterize the types of sensitive trigger data accessible to TAPs,
we analyzed the IFTTT dataset mentioned in \cite{mi2017empirical}, by mapping
each of its 320,000 IFTTT rules to one of the three trigger sensitivity levels defined by Bastys
et al.~\cite{DBLP:conf/ccs/BastysBS18} --- 
public, private, and time-sensitive. Private triggers contain information like
emails and calendar events, whereas public triggers contain information like
news and weather reports. The time-sensitivity level means that private
information exists in the availability of the trigger message. For example,
considering the rule ``IF I leave home, THEN turn off the WiFi,'' the TAP will
learn whether the user leaves home depending on whether it receives a message
from the trigger service. \figref{fig:classify} shows a breakdown of sensitive
trigger data according to how frequently they are used.

We observe that although a significant percentage (15\%) of triggers and action
APIs supported by IFTTT are time-sensitive, in reality, they are rarely used ---
only 0.8\% of all available rules in IFTTT (or 0.9\% of all installed rules) use
a time-sensitive trigger. We also observe that, although there are fewer private
triggers than public ones, private triggers are most frequently used --- 61\% of
all installed rules contain a private trigger API. \change{These APIs return private information like emails, messages, location traces, photos, sensitive files, medication lists, health information, etc. Thus, we design \name to protect the vast majority of private trigger information that people actually use in real-world rules.} We do not currently provide
confidentiality for time-sensitive information, but we outline possible
approaches using standard techniques like cover traffic in
Section~\ref{sec:discussion}.



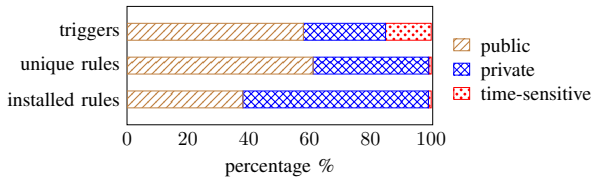
\begin{figure}[t]
\centering
  \begin{tikzpicture}[scale=0.75]
    \begin{axis}[
      xbar stacked,
      legend style={
        legend columns=1,
        at={(1.3, 0.8)},
        anchor=north,
        draw=none
      },  
      legend cell align={left},
      ytick=data,
      y=-0.6cm,
      enlarge y limits={abs=0.45cm},
      bar width=0.3cm,
      xmin=0,
      xmax=100.5,
      symbolic y coords={triggers, {unique rules}, {installed rules}},
      xlabel={percentage \%},
      ]
      \addplot+[xbar, pattern=north east lines, pattern color = brown, draw=brown]  coordinates {(58,triggers) (61.2,{unique rules}) 
        (38.1,{installed rules})};
      \addplot+[xbar, pattern=crosshatch, pattern color = blue, draw=blue]  coordinates {(27,triggers) (38,unique rules) 
        (61,installed rules)};
      \addplot+[xbar, pattern=crosshatch dots, pattern color = red, draw=red]  coordinates {(15,triggers) (0.8,unique rules)
        (0.9,installed rules)};
      \legend{~public, ~private, ~time-sensitive}
    \end{axis}
  \end{tikzpicture}
  \caption{Breakdown of triggers, rules, and installed rules in IFTTT based on their sensitivity levels.
  }
  \label{fig:classify}
\end{figure}

\paragraph{Operations on trigger data.} IFTTT allows users to express computation on trigger data using \emph{filter
  code} --- small snippets of TypeScript with some restrictions (e.g., no I/O operations)~\cite{ifttt-code}. Zapier rules contain two components: \emph{filters} that compute a predicate on the trigger
data, 
and \emph{formatters} that modify the trigger data. Multiple filters and
formatters can be chained together.

To understand the common operations in IFTTT, we again used the dataset of Mi et al.~\cite{mi2017empirical}. 
We selected the $500$ most popular rules (based on user installation count) that are
connected to private trigger APIs. Unfortunately, a challenge is that filter
codes for IFTTT rules are not public. We therefore manually approximated the
filter code for these rules by (1) estimating the functionality of each rule
based on their title and description, (2) examining the corresponding
trigger/action APIs, and (3) deducing the operations that are required to
convert trigger fields to action
fields.


We also crawled the Zapier website for one day in October 2019 and collected all the
publicly available rules that require computations on trigger data~\cite{zapier:filter, zapier:formatter}.  
We collected a total of 378 rules and extracted the
operations used in those rules.

The operations we found in IFTTT and Zapier are shown
in~\figref{fig:functions}. Current garbled
circuit libraries support a majority of these operations natively. The main
challenge is string operations, for which we contribute a novel technique to convert deterministic
finite automata into Boolean circuits (\secref{sec:boolfunc}).




\begin{figure}[t]
\centering\tabfontsize
\begin{tabular} {l p{2.3cm} p{4.6cm}} 
  \toprule
  \textbf{Type} & \multicolumn{1}{l}{\textbf{Operation}} & \multicolumn{1}{l}{\textbf{Description}} \\
  \toprule
  \multirow{3}{*}{Bool}
                & \texttt{x | a} & \texttt{x} OR \texttt{y}  \\
                & \texttt{x \& a} & \texttt{x} AND \texttt{y} \\
                & \texttt{! x} & NOT \texttt{x}  \\
  \midrule
  \multirow{4}{*}{Num}
                & \texttt{x < n} & Is \texttt{x} \emph{less than} \texttt{n}?  \\
                & \texttt{x > n} & Is \texttt{x} \emph{greater than}  \texttt{n}?  \\
                & \texttt{x.mathop(n)} & Basic math ops. ($+,-,\times,\div)$ 
                \\
                & \texttt{x.format()} & Format \texttt{x} into a string \\
  \midrule
  \multirow{13}{*}{Str}
                & \texttt{x == s} & Does \texttt{x} \emph{exactly match} the string \texttt{s} \\
                & \texttt{x.contain(s)} & Does \texttt{x} \emph{contain} the string \texttt{s} \\
                & \texttt{x.startwith(s)} & Does \texttt{x} \emph{start with} the string \texttt{s}  \\
                & \texttt{x.endwith(s)} & Does \texttt{x} \emph{end with} the string \texttt{s}  \\
                & \texttt{x.split(d, i)}  & Split \texttt{x} using delimiter string \texttt{d} and select the \texttt{i}-th substring \\
                & \texttt{x.replace(s, t)} & Replace all occurrences of \texttt{s} in \texttt{x} with \texttt{t}    \\
                & \texttt{x.to\_lowercase()}  & Convert all characters in  \texttt{x} to lowercase \\
                & \texttt{x.truncate(n)}  & Truncate  \texttt{x} to size  \texttt{n} \\
                & \texttt{x.extract\_phone()}  & Extract the first phone number found in \texttt{x}   \\
                & \texttt{x.extract\_email()}  & Extract the first email address found in \texttt{x}   \\
                & \texttt{x.strip\_html()}  & Remove all HTML tags in \texttt{x}  \\
                & \texttt{x.html2markdown()}  & Convert all HTML tags in \texttt{x} to Markdown  \\
                & \texttt{m.lookup(x)}  & Look up the value for the key \texttt{x} in a user-provided map \texttt{m} \\
  \midrule
  \multirow{2}{*}{Any}
                & \texttt{x == null} & Does \texttt{x} \emph{exist}? \\
                & \texttt{x.default(y)}  & Set value of \texttt{x} to \texttt{y} if it does not exist \\
  \bottomrule
\end{tabular}

\caption[]{Operations used in top 500 IFTTT rules with private triggers and all Zapier's function-dependent rules.
}
\label{fig:functions}
\end{figure}

\paragraph{Execution model of trigger-action systems.}
Based on our survey of IFTTT and Zapier, we derive an abstract model of these
trigger-compute-action rules. During rule setup on the client, the user typically
specifies two functions --- a \textit{predicate} $\fone$, and a
\emph{transformation} $\ftwo$. These functions take the trigger data
and some additional user-provided constants as input. 
The predicate function $\fone$ tests the trigger data for a condition to
determine whether TAP should contact AS. The output
of $\fone$ is either $\true$ or
$\false$. 
The transformation function $\ftwo$ modifies the trigger data before sending the
result to AS.  
Both $\fone$ and $\ftwo$ run inside the cloud-based TAP.

Let $\dt\in\calX$ be the part of the trigger data on which TAP performs some
computation, and $\da\in\calY$ be the action data TAP sends to AS, where $\calX$
and $\calY$ are the domains of the trigger and action data, respectively.  Both
$\dt$ and $\da$ can be data structures that contain multiple fields.  We find
that TAPs do not modify some fields of trigger data such as large media 
files, but only forward them to AS.  We denote such trigger data as payload
$\payload$.  Let $\cone, \ctwo\in\calC$ be the two user-provided constants for
the functions $\fone$ and $\ftwo$, where $\calC$ is the domain of the constants.
On receiving $(\dt, \dtp)$ from TS, TAP executes
\bnm \text{``if }\fone(\dt, \cone)=\true
\text{, then send }(\ftwo(\dt,\ctwo),\dtp) \text{ to AS''} \enm
For simplicity, we
assume the domains of $\fone$ and $\ftwo$ to be the
same. 
So, $\fone\Colon\calX\times\calC\to \{\true, \false\}$, and
$\ftwo\Colon\calX\times\calC\to\calY$.

TAPs operate in two modes: 
(1) \emph{polling mode}, where TAP contacts TS at a predefined frequency; (2)
\emph{push mode}, where TS sends a message to TAP when an event occurs. While
our protocol will work with both models, we assume the push model in this paper as
it is more efficient in general.


\paragraph{Example rule.}
We show how our abstract model can instantiate our previous example rule:
``IF I receive an email containing the word `confidential', then send me an SMS.''
The SMS should contain the address of the sender and the email's subject. 
Assume that TAP provides an operation to search over strings, called
\texttt{contain}.
The user sets up a rule by choosing its email provider as the trigger service,
that sends a copy of every new email to TAP.  The action service is an SMS
provider that sends SMS to a user-provided number.  The user then specifies the
\texttt{contain} function to check for the string
$\cone=$``\textsf{confidential}'' on the email's subject line, $x$.
The transformation function $\ftwo$ creates the required data structure to send
the SMS, for example, setting the recipient address as the user-provided phone
number $\ctwo$ and the message body as the concatenation of the sender's address
and the subject.


\section{Design Considerations for Providing Data Confidentiality in
  Trigger-Action Systems}
\label{sec:towards}
Our goal is to protect the confidentiality of private data involved in
trigger-action rules even if they are run on a malicious cloud-based TAP. In
this section, we discuss our threat model, define our security and functionality
goals, and explore the design space.

\subsection{Threat Model and Functionality Goals}\label{sec:threat}

Fernandes et al.~\cite{fernandes2018decentralized} first noted the security and
privacy issues of a compromised \TAP and the related attacker motivations. We
adopt the same attacker model --- \TAP is \emph{malicious}. Specifically, the
attacker: (1) can monitor communications between \TAP and the trigger/action
services; (2) can arbitrarily deviate from the communication protocol by
manipulating, delaying, or dropping the messages; \change{(3) can modify TAP's internal storage and code that includes manipulating and deleting garbled circuits;}
(4) knows API details of
trigger and action services; and (5) knows the functions that are being
evaluated on \TAP. As we use cryptographic techniques for our security
guarantees, we assume that the attacker is computationally bounded.

We assume that the end-point services (trigger and action services) like Samsung
SmartThings, Google Calendar, etc. are \emph{semi-honest} --- they will follow
the protocol as specified, but try to glean more information than 
\change{what they are entitled to know}. This is in line with the trust model used by current
TAPs. Also, if they are compromised, then the attacker can achieve its goals of
accessing and manipulating user data independently of the trigger-action system.
We also assume that \change{TAP is not colluding with TS or AS}. 
\change{As discussed in~\secref{sec:background}, third-parties own a large majority of trigger and action services and thus collusion with TAP is unlikely (for example, there is no incentive for LG or Google to collude with IFTTT to reduce the security of their users). Enforcement of the non-collusion condition can also be done via legal affidavits \cite{legal, CryptE} or techniques that involve using a trusted mediator who monitors the communications between the parties  \cite{non-collusion1,non-collusion2}.}


Finally, we assume that the user trusts their client device. We observe that the attacker is motivated to compromise \TAP
because it will simultaneously be able to attack all users of the platform. An
attack on the client device is not scalable  to
all users easily, and therefore, is less attractive.

\paragraph{Security goals.}  Under this threat model, we want two security
properties for a trigger-action system:
\\[2pt]
\emph{Privacy:} Each party should not learn other parties' data in a trigger-action rule. Specifically, 
    \TAP should not learn the trigger data
  $(\dt, \dtp)$, user-provided constants $(\cone, \ctwo)$, and results of the
  computation (beyond what they already know from the definitions of the
  functions);
  the trigger service (\TS) should not learn the user-provided constants ($\cone, \ctwo$);
  the action service (\AS) should not learn the trigger data $\dt$ or user-provided constants ($\cone, \ctwo$) beyond what is revealed to it after rule execution. Additionally, \AS should not learn the output of transformation function $\ftwo$ or payload $\dtp$ when the predicate function $\fone$ evaluates to $\false$.
\\[2pt]
\emph{Integrity:} \change{The attacker should not be able to modify any
    computations on private trigger data without being detected by  AS.}
  That is to say, \TAP should not be able to trick  \AS into acting on
  illegitimate action data, such as delayed, replayed, or tampered messages that
  are not the result of proper evaluation of the rule. \AS only accepts valid messages
  $\da = \ftwo(\dt, \ctwo)$, where $\dt$ is sent by \TS within the last $\delay$
  seconds (a configurable parameter).

\paragraph{Security non-goals.} Denial of service is outside our scope. A compromised \TAP can indeed drop all messages it receives from \TS and not transmit any message to \AS. \change{Metadata and side-channel attacks are also outside our scope. For example, even if messages are encrypted, the compromised \TAP can observe the timing of messages that arrive from a trigger service or go to an action service. Coupled with semantic knowledge about the services, this might enable the attacker to determine the sensitive data in the rule even if it is encrypted. As discussed in~\secref{sec:supported-functions}, this involves time-sensitive rules which are less used frequently in practice.
    \stap protects the vast majority of sensitive trigger data for which encryption achieves strong security properties. 
    \secref{sec:discussion} outlines standard approaches to protect metadata that we leave as future work.}

\paragraph{Functionality goals.} We want to achieve the security goals while respecting the following functionality goals:
\emph{(1) RESTful API for end-point services}. The end services should be able to design their APIs
independently of each other, as they do currently. These APIs should be RESTful, have minimal computational overhead beyond running the API itself, and do not need to store data or state specific to different trigger-action rules.
\textit{(2) Maintain trigger-compute-action paradigm.} The design should run existing user-created rules without any changes and should maintain the key architectural aspects of current trigger-action systems. Notably, the rules should execute without requiring the client device to be online.




\subsection{Design Space Exploration}\label{sec:design-explore}

We explore a few potential solutions occupying different points in the design space and discuss why they do not meet our functionality or security requirements.

\paragraph{Computation at the edges.} The trigger service can run a user-supplied function
over its private data, encrypt the result, and forward that to \TAP. \change{However, the trigger service has to support an execution infrastructure similar to AWS Lambda, significantly increasing the complexity and overhead of such services and exposing them to additional security risk due to executing third-party code. 
} Furthermore, sensitive data in user-supplied constants
($\cone$, $\ctwo$) will be exposed in plaintext to the trigger
service. 
\change{For example, consider rule R7 from~\figref{fig:rules-exp}, which converts Slack mentions to Asana tasks (a project management tool). It requires users to provide a lookup table of project names. These are sensitive information that should not be revealed to Slack.} Computation can also be moved to the action service, but the same issues exist there as well.



\paragraph{Secure hardware.} It is possible to use hardware-based trusted execution
environments (TEEs) or hardware security modules (HSMs) for computing
the trigger data on TAP, while preserving  confidentiality~\cite{Zavalyshyn+:MobiQuitous20,Schoettler+:Walnut}. Yet besides requiring hardware changes to the TAP servers, current TEEs
suffer from fundamental security design
issues~\cite{foreshadow,plundervolt,sgxpectre}.


\paragraph{Homomorphic encryption of the trigger data.} During rule setup, the
client can specify a symmetric key between the trigger and action service. The
trigger service encrypts its data using this key before sending it to \TAP. This
will provide trigger data confidentiality and allow the \TAP to compute directly
on the encrypted data. However, only specialized schemes like linear homomorphic
encryption and ``somewhat'' homomorphic encryption are practical~\cite{SWHE},
thus limiting expressivity.
\change{For reference, TFHE~\cite{tfhe}, a state-of-the-art library for fully homomorphic encryption, takes 4.45~seconds to compute an addition circuit, which is 3 orders of magnitude slower than our system as evaluated in~\secref{sec:eval}.}
Additionally, protection against a malicious \TAP
would require zero-knowledge proofs~\cite{GMW} of computation that would further
reduce efficiency.


\paragraph{Off-the-shelf secure multi-party computation.}  Secure multi-party
computation (SMC) protocols allow multiple
distrusting parties to compute a function over their private
inputs~\cite{yao1982protocols}. However, efficient off-the-shelf SMC protocols
do not fit our threat model --- TAP is malicious, or architectural
requirements --- needing TC, TS, AS, and TAP to participate in a multi-round
protocol during rule execution. Therefore, we adopt a core
primitive of SMCs --- garbled circuits --- and modify it to our setting.

\paragraph{Secret sharing based SMC.}  Secret sharing is an alternative to
garbled circuits for doing SMC. However, secret sharing-based protocols require
 intensive multi-round communication (e.g., for evaluating multiplication gates).
Additionally, in such protocols every party has to do an equal amount of work,
which will require invasive architectural changes to TS and AS. This violates
our functionality goal. Finally, the malicious versions of these protocols
are not efficient.

\section{Design of Encrypted Trigger-Action Platform}
\label{sec:design}

In this section, we discuss \name's core protocols and analyze how we specialize
garbled circuits to trigger-action systems. A high-level overview of \name is
shown in~\figref{fig:arch_overview}, and the pseudocode is given in~\figref{fig:algo}.  Like a typical trigger-action system
in~\figref{fig:tas}, \name has four components: trusted client's device ($\TC$),
trigger service ($\TS$), action service ($\AS$), and a trigger-action platform
($\TAP$). We describe below how our design modifies these four components while
maintaining the trigger-compute-action paradigm.


\paragraph{Decentralized trust model.}
In the current trigger-action system design, users place all trust within a
centralized cloud-based TAP.  This design leaves open a large-scale security and
privacy risk --- a single compromise of the TAP will simultaneously compromise
all users. To avoid this issue, \name borrows a design element from
DTAP~\cite{fernandes2018decentralized} and designates the user's client device
(smartphone) as the root of trust. Each user \textit{only} trusts their own
smartphone and uses it to program trigger-compute-action rules. As the \name
protocols are open-source, we envision a community of developers building client
apps, much like we have apps for open protocols like SFTP, Telnet, etc. Thus,
the \name cloud component and the client app are built and controlled by 
different entities. Therefore, the client app can still be trusted, even when the TAP
is
compromised.  
\name bootstraps its guarantees on top of this model. In \name, the trusted
client (\TC) is beyond just an interface --- it stores some state (as we
describe below) that is key to its operation.

\subsection{Rule Setup (occurs on trusted client)}
Like in existing trigger action systems, the user can configure a
trigger-compute-action rule on the trusted client app
(\TC) 
using its click-through interface. The user selects a trigger in a trigger
service (\TS), a predicate $\fone$, a data transformation over the trigger data
$\ftwo$, and an action in an action service (\AS). The user also specifies any
constants $c$ if required.

\TC sends the rule descriptions to \TAP and helps the \TAP negotiate OAuth tokens
with TS/AS required for running the rule.
In \name, unlike existing TAPs, \TC shares with \TS and \AS two
uniformly-generated secret keys $\sects$ and $\secas$, upon successful
authorizations.  The key $\sects$ and $\secas$ are tied to the specific trigger
and action API for this user in \TS and \AS\footnotemark.  If a prior rule has
already been set up with the same trigger or action API, then the corresponding OAuth authorization can be skipped and \TC will reuse the
previously generated $\sects$ or $\secas$.  Once the rule is setup, \TS and \AS
store the shared key materials; \TAP stores the OAuth tokens; and \TC stores the
rule ($\fone, \ftwo$), the keys ($\sects, \secas$), and the
constants ($\cone, \ctwo$) provided by the user for the rule.

\footnotetext{For better usability, current TAPs only acquire one OAuth token
  per service that can access all APIs in
  it~\cite{fernandes2018decentralized}. \name can adapt to this model by
  exchanging a service-level key $\sectservice, \secaservice$, and derive the
  API-level keys $\sects, \secas$ from the hash value of
  $\sectservice, \secaservice$ and API URL, as required.}



\begin{figure}[t]
\centerline{\includegraphics[width=0.5\textwidth]{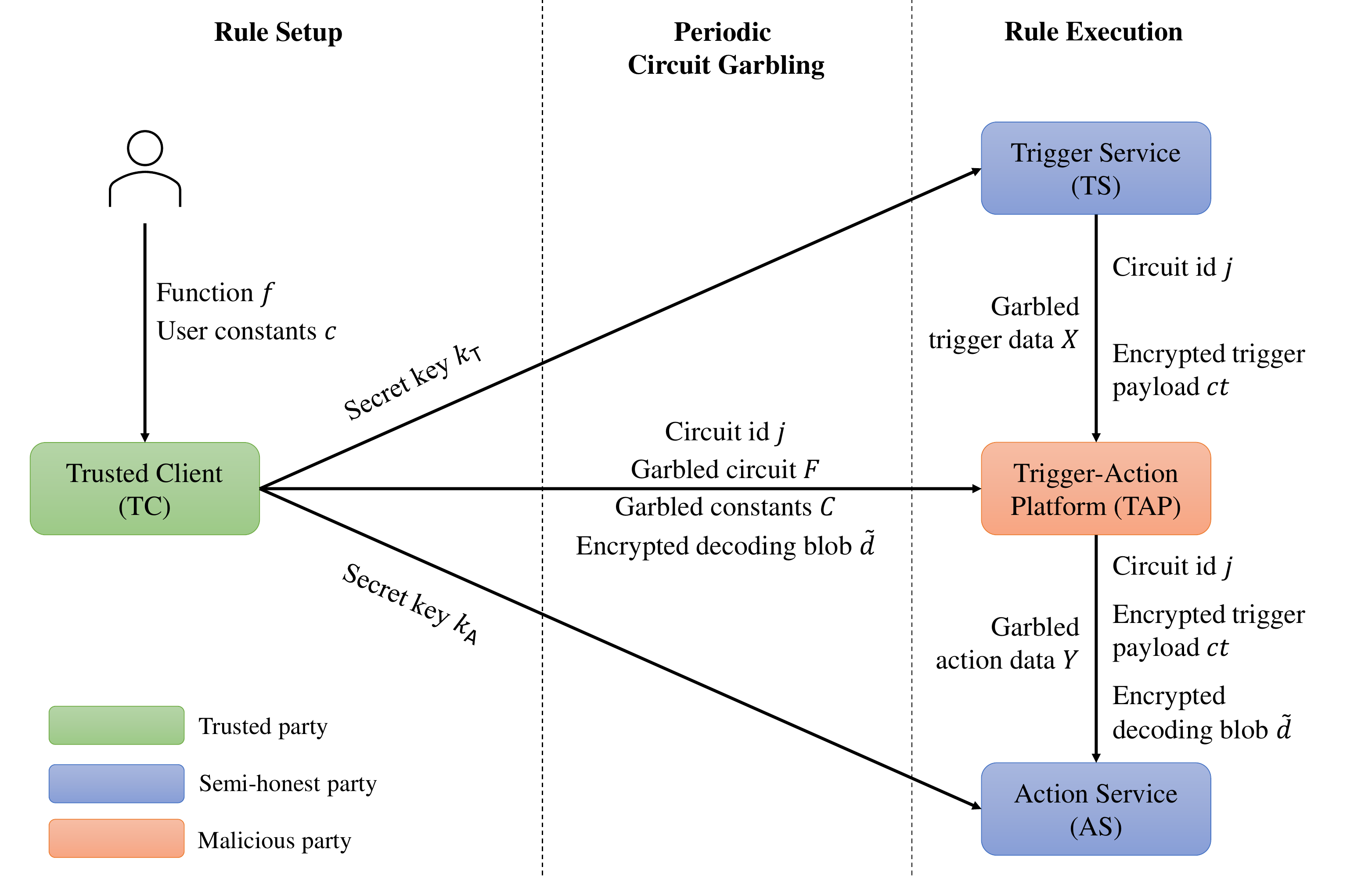}}
\caption{Overview of \stap.}
\label{fig:arch_overview}
\end{figure}

\begin{figure}[t]
\hspace{-1em}
  \hfpagess{0.24}{0.215}{ \setstretch{1.1} \gamesfontsize
    \underline{\cktgarbling$\left((f, \const), (\sects, \secas, j)\right)$}:\\[2pt]
    $\eis \gets \H(\sects \| j \| 0) $\\
    $\eir \gets \H(\sects \| j \| 1) \lor 0^{\keylen-1}1 $\\
    $\payloadk \gets \H(\sects \| j \| 2) $\\
    $\ei \gets (\eis, \eir)$\\
    $(F, \wlabel^{w_0}_0, \ldots, \wlabel^{w_m}_0) \gets \Gb'(\ei, f)$\\
    $d' \gets (\textsf{lsb}(\wlabel^{w_1}_0), \ldots, \textsf{lsb}(\wlabel^{w_m}_0))$\\
    $h \gets \H(\wlabel^{w_1}_0 \concat \ldots \concat \wlabel^{w_m}_0)$\\
    $\encs \getsr \skeenc(L^{w_0}_1 \oplus \secas, (j, \payloadk, \eir, d', h)) $\\
    $\ench\gets\hmac_{\secas}(j\|L^{w_0}_0)$\\
    $\encd \gets (\encs,\ench)$\\
    $C \gets \En(\ei, \const)$\\
    Set $j = j + 1$\\
    Return $j, F, C, \encd$\\

    \underline{\tsexec$\left((\dt, \payload), (\sects, j)\right)$:}\\[2pt]
    $\eis \gets \H(\sects \| j \| 0)$\\
    $\eir \gets \H(\sects \| j \| 1) \lor 0^{\keylen-1}1  $\\
    $\payloadk \gets \H(\sects \| j \| 2) $\\
    $\LT \gets \En\left((\eis,\eir), \dt\right)$\\
    $t \gets \textsf{CurrentTime()} $\\
    $ct \getsr \skeenc(\payloadk, (t, \payload))$\\
    Set $j = j + 1$\\
    $\return j, \LT, ct$
   }
  {\gamesfontsize
    \underline{\tapexec$\left((j,X,ct), (F,C,\encd)\right)$:}\\[2pt]
    $Y \gets \Ev(F, (\LT, C))$\\
    $\return j, \LA, ct,\encd$\\

    \underline{\asexec$\left((j, \LA, ct,\encd), \secas\right)$:}\\[2pt]
    Parse $\LA$ as $(\wlabel^{w_0}, \ldots, \wlabel^{w_m})$\\
    $(\encs,\ench) \gets \encd $\\
    $z\gets \skedec(\wlabel^{w_0} \oplus \secas, \encs)$\\
    If $z = \bot$ then \\
    \mytab $\ench' \gets \hmac_{\secas}(j\|\wlabel^{w_0})$\\
    \mytab If $\ench' \ne \ench$ then \; Return $\bot$ \\
    \mytab Else \; Return $\false$\\
    $(j', \payloadk, \eir, d', h) \gets z$\\
    If $j\ne j'$  then \; Return $\bot$\\
    $y \gets \De\left(d', (\wlabel^{w_1}, \ldots, \wlabel^{w_m})\right)$\\
    $g\gets \bot$\\
    For $i\gets 1$ to $m$ do\\
    \mytab If $y_i = 0\;$ then $\;g\gets g\concat\wlabel^{w_i}$ \\
    \mytab Else $\;g\gets g \concat (\wlabel^{w_i}\oplus\eir)$\\
    $h' \gets \H(g)$\\
    $(t, \payload) \gets \skedec(\payloadk, ct)$\\
    $t'\gets \textsf{CurrentTime}()$\\
    If $t'> t+\delay$ or $h \ne h'$  then\\ 
    $\mytab\return~\bot$\\
    $\return y, \payload$
  }
  
  \caption{Circuit generation and rule execution protocols for $\stap$.
    $\wlabel^{w_0}_1$ denotes the $\true$ label for the first output wire $w_0$,     $\wlabel^{w_0}_1 = \wlabel^{w_0}_0 \oplus \eir$;
    $\delay$ is a threshold parameter used to ensure the freshness of a
    trigger. \textsf{CktGarbling} is run by $\TC$
    asynchronous to the actual rule execution.  The remaining three functions
    are run by \TS, \TAP, and \AS during rule execution.
}
  \label{fig:algo}
\end{figure}

\subsection{Circuit Garbling (periodic, occurs on trusted client)}
Once the user creates a new rule,  \TC has to generate garbled circuit to enable
secure evaluation of the functions on the (untrusted) \TAP. \TC
generates garbled circuits corresponding to $\fone$ and $\ftwo$ and the
associated encoding/decoding blobs. It uses the encoding blob to obtain the
garbled labels for user-supplied constants. The
decoding blob allows \AS to decode the garbled outputs and to decrypt the payload. 
To ensure \TAP does not learn or tamper
with the decoding blob, \TC encrypts it using $\secas$.  \TC sends the garbled circuits,
encoded constants, and encrypted decoding blob to \TAP. 
\TC identifies each instance
of the garbled circuit using a monotonically increasing counter $j$.
The circuit id $j$ is initialized to zero if this is the first rule where the user uses the connected trigger API;
otherwise, \TC queries \TAP for the circuit id that the connected trigger API is currently using.
As garbled
circuits cannot be reused, \TC periodically repeats the above process.

Although \TC needs to transmit the garbled circuits and related information
prior to rule execution (\figref{fig:arch_overview}), we design \name such that
\TC does not need to be online during execution. \TC generates and transmits GCs
in batches at times when the smartphone is not being used (e.g., when charging
at night).  Our evaluation (\secref{sec:macro}) demonstrates that transmitting
sufficiently many garbled circuits for a day generally takes less bandwidth than
backing up a 1-minute Full HD video to a cloud drive. This achieves our design
principle of keeping the client device offline during rule execution.

Note that in our setting, the generator of the garbled circuit is the smartphone client --- a trusted entity. This is a key insight and design element that is possible due to the nature of our setting. This allows \name to use efficient semi-honest implementations of garbled circuits and achieve security in the presence of a malicious \TAP.

\paraheading{\em Cryptographic Details.\;}
Without loss of generality, we assume $\fone\Colon\zo^n\times\zo^n\to\zo$ and
$\ftwo\Colon\zo^n\times\zo^n\to\zo^m$. For notational simplicity, we denote
$f\Colon\zo^n\times\zo^{2n}\to\zo^{m+1}$, such that
$f(\dt, \const) = \fone(x, \cone)\|\ftwo(x, \ctwo), \const=(\cone, \ctwo)
\in\zo^{2n}$. Additionally, let $\H: \zo^*\to\zo^{\kappa}$ denote a
cryptographic hash function. The pseudocode of the circuit garbling is given by
the $\textsf{CktGarbling}$ function in~\figref{fig:algo}.

\paraheading{Encoding blob.}
The encoding blob contains the information required to encode the trigger data and encrypt
the trigger payload. It can be derived from the key $\sects$ and the
garbled circuit id $j$.  \TC 
generates three bitstrings $(\eis, \eir, \payloadk) \in \zo^{3\keylen}, $ using the
hash of $\sects \concat j$. 
The $\false$ labels of the input wires (as
described below) are generated using a $\H$ with $\eis$ as the random seed, and
$\eir$ is used as a global offset for the standard free-XOR
optimization~\cite{kolesnikov2008improved}. The least significant bit of $\eir$ is set to $1$
to enable the standard point-and-permute optimization~\cite{BMR90, zahur2015two}. Thus $\ei=(\eis, \eir)$ constitutes the encoding
information used for the garbling scheme's encoding function
(\En). 
The key $\payloadk$ is used to protect the payload
data~$\payload$. 

\paraheading{Garbled circuit.} To generate the garbled circuit $F$ for function
$f$, the labels for every input wire $w$ are computed as
$\wlabel^w_0 = \H(\eis\| w)$ and $\wlabel^w_1 = \wlabel^w_0\oplus \eir$
(assuming wire index $w$ is a fixed-length bitstring).  The rest of the
computation (generating labels of the non-input wires and garbling gates)
proceeds as per standard techniques with optimizations, such as
row-reduction~\cite{NPS99} or half-gate~\cite{zahur2015two}.



%

\paraheading{Encrypted decoding blob.}  The decoding blob consists of information
necessary for \AS to decode the labels of output wires (that correspond to the
action data $y$) and to decrypt the payload. Let the output wires be
$(w_0, w_1, \ldots, w_m)$, where $w_0$ corresponds to the output wire of
$\fone$, and the remaining $m$ wires correspond to those of $\ftwo$. Following
standard practice~\cite{BMR90}, the decoding information $d$ contains the least significant bits (lsb)
of the $\false$ label of each output wire
$\left(\lsb(\wlabel^{w_0}_0), \ldots, \lsb(\wlabel^{w_m}_0)\right)$. In
\name, decoding information is slightly modified. First, the first bit, $\lsb(\wlabel^{w_0}_0)$, of $d$ is 
dropped to create $d'$. 
Second, the hash of all the $\false$ labels of $\ftwo$'s output wires
$h \gets \H(\wlabel^{w_1}_0 \concat \ldots \concat \wlabel^{w_m}_0)$ is computed. 
Third, a decoding blob is created using $d'$, $h$, the payload key $\payloadk$, the XOR offset $\eir$,
and the current circuit id $j$. 
Next, the whole blob is encrypted using a symmetric-key encryption scheme $\skeenc$
with a key derived from both $\secas$ (the secret key shared with \AS) and
$\wlabel_1^{w_0}$ (the $\true$ label of $\fone$'s output wire $w_0$) to obtain
$\encs \getsr \skeenc(L^{w_0}_1 \oplus \secas, (j, \payloadk, \eir, d',
h))$. Additionally, an HMAC~\cite{BellareHMAC} of the $\false$ label of predicate
$\fone$ is computed using $\secas$ as $\ench\gets\hmac_{\secas}(j\|\wlabel^{w_0}_0)$. 
We use $\encd$ to denote the tuple $(\encs,\ench)$. We explain the rationale
behind these changes in~\secref{sec:design-rationale}.


\paraheading{Encoded user constants.} Using the encoding information $\ei$, \TC
computes the labels for constants $c$ as
$C \gets \En(\ei,\const)$. 

To accommodate the above customization, we derandomize the garbling function
$\Gb$ to $\Gb'$ that takes an encoding information $e$ as an input and returns
the garbled circuit $F$, as well as the $\false$ labels of every output
wire. 
\TC sends $(j, F, C,\encd)$ to
\TAP and increments the circuit id $j$ by 1.

\subsection{Rule Execution (occurs on TAP; does not involve \TC)}
When new trigger data is available for a trigger API, \TS will garble the input
data and encrypt any payload data, using the encoding blob it computes from
$\sects$ and circuit id $j$ (which is initialized to 0 when the API is first
called). It then transmits the ciphertexts to \TAP, which will lookup any rules
that are connected to the trigger API (and user) and run the associated garbled
circuits. \TAP finally transmits the output of the evaluation (garbled action
data) and the encrypted decoding blob to the corresponding API in \AS, which can
decode to the plaintext result using $\secas$ (\figref{fig:arch_overview}).

\TS and \AS only perform simple encoding and decoding of data --- fixed  functionality independent of the trigger-action rule semantics, thus maintaining their RESTful nature. We believe that \TS and \AS are well-motivated to support these additional operations, in exchange for enhanced security. Indeed, current end-point services are concerned about the privacy of user data. For example, GMail recently removed their IFTTT triggers citing security and privacy concerns~\cite{gmail-ifttt-removal}. 

In our setting, the full evaluation of the garbled circuit is split between the untrusted \TAP that executes the circuit to produce garbled output labels and the semi-honest \AS that decodes the plaintext result from the labels. This, in combination with the trusted generator, allows \name to efficiently achieve the execution authenticity property of GCs using a hash function (Section~\ref{sec:design-rationale}), even when \TAP itself is malicious. 
We omit the standard OAuth steps that occur during execution, which the reader can refer to~\cite{ifttt-oauth} for details.

\paragraph{\em Cryptographic Details.}
\TS's operations in the rule execution phase is function $\tsexec$
in~\figref{fig:algo}.  \TS recomputes the encoding information $\ei=(\eis,\eir)$
and the payload key $\payloadk$ from $\sects$ and $j$.  It then encodes the
trigger data $\dt$ using the garbling scheme's
encoding function, producing $X\gets \En(\ei,\dt)$, and encrypts the payload
$\payload$ under a symmetric-key encryption scheme with the key $\payloadk$ to
compute $ct\getsr\skeenc(\payloadk, (t, \payload))$ where $\ts$ is the current
timestamp. Finally, \TS forwards the message $(j, X, ct)$ to TAP and increments
$j$ by 1.

Upon receiving a trigger message $(j, X, ct)$, \TAP retrieves the corresponding
garbled circuit $F$, garbled constants $C$, and the encrypted decoding blob
$\encd$ using the trigger API and the circuit id $j$. Next, \TAP evaluates $F$
to obtain the garbled action data $Y\gets \Ev(F,(X,C))$ and forwards the tuple
$(j, Y, ct, \encd)$ to \AS. Function \textsf{TAPExec} in \figref{fig:algo}
depicts this process.

After receiving a message from \TAP, \AS decrypts $\encd$ to obtain the decoding
information, which will succeed only when $\fone$ evaluates to $\true$ (i.e.
$\wlabel^{w_0} = \wlabel_1^{w_0}$). If \AS is able to decrypt the decoding blob,
it uses $(d', \payloadk)$ to obtain the final output
$(\ftwo(x,\ctwo), \payload)$ in plaintext.  \AS would terminate if the message
from \TAP is malformed (i.e., hash of labels is inconsistent or decryption
fails) or stale (i.e., trigger timestamp is old). The function \textsf{ASExec}
in \figref{fig:algo} depicts this process.



\subsection{\change{Rationale for Novel GC Protocol \& Security Analysis}}
\label{sec:design-rationale}
\change{\name adopts a customized GC-based protocol tailored to the needs of trigger-action platforms. This protocol is novel in the following ways: (1) By leveraging the structure and threat model of trigger-action systems, we can use efficient semi-honest implementations of GCs to obtain security against a malicious evaluator; (2) \name supports fixed-length string operations including matching, extraction, and replacement --- common operations in trigger-action programs --- using Boolean circuits only; (3) \name contributes an efficient technique to ensure authenticity on the evaluator's output (i.e., \TAP) that requires only two hashes instead of the existing standard approach that requires hashes for $\true$ and $\false$ labels for every output wire. 
}


  Our setting has four parties: \TC generates the garbled circuit via $\Gb'$
  and then, both \TC and \TS use $\En(e,\cdot)$ to encode their respective
  inputs.   On the other hand,
  \TAP evaluates the garbled circuit using $\Ev(F,\cdot)$ while \AS decodes the
  plaintext output using $\De(d',\cdot)$. \change{Thus, \TC and \TS jointly act as the ``generator'', and \TAP and \AS jointly
  emulate the role of the ``evaluator'' of a two-party computation setting.}  The evaluators (\TAP and \AS) in our
  setting do not have any private input, therefore, 
  \name does not require any oblivious transfers. 
  Trust assumptions of the constituent parties of the generators and
  evaluators are asymmetric. Among the generators, \TC is fully trusted and
   \TS is semi-honest; among the evaluators, \AS is semi-honest and 
  \TAP is fully malicious. \change{Recall, \TS and \AS do not collude with TAP.} (See~\secref{sec:threat} for the motivations behind these trust assumptions.)

Next, we highlight the changes we introduce in two-party GC protocol and the
rationale behind those changes. We formally prove all security properties of our protocol in~\appref{sec:sec_analysis}.

\noindent \textbf{(1)}    \TC generates the encoding information deterministically from the shared
   secret key $\sects$ and the circuit id $j$, so that $\TS$ can also generate
   it without any communication with \TC during rule execution.  This
   achieves our design goal of ensuring that \TC can be offline during rule
   execution.  We note that this change does not violate the input privacy
   guarantees of the GC (see \thref{thm:TAP1}, \ref{thm:TS},
   \ref{thm:AS1}, and \ref{thm:AS2}).

\noindent \textbf{(2)}
   Recall that the decoding blob (which contains information to decode garbled action data and to decrypt payload)
   is 
   encrypted using the bit-wise XOR of $\secas$ and $\wlabel_1^{w_0}$ as the
   key. Thus, \TAP cannot learn the decoding blob (it does not have $\secas$,
   \thref{thm:TAP1}). Only \AS can successfully decrypt $\encs$ if it gets the
   $\true$ label of the output wire of $\fone$, $\wlabel_1^{w_0}$, from \TAP; which 
   can happen only when the predicate $\fone(x)$ evaluates to true. This
   meets our privacy requirement that \AS should not learn $\ftwo(x,\ctwo)$ or
   $\dtp$ when $\fone(x,\cone)=\false$. We formally prove this
   in~\thref{thm:AS1} and \ref{thm:AS2}.

\noindent \textbf{(3)} \name ensures that the malicious TAP (evaluator) cannot tamper with the
   results of evaluation. To achieve this we add the following information to the decoding blob: 
   $h=\H(\wlabel^{w_1}_0 \concat \ldots \concat \wlabel^{w_m}_0)$, the XOR
   offset $\eir$, and $\H(\wlabel^{w_0}_0)$. Standard techniques to achieve this
   property require the hashes of both $\true$ and $\false$ labels for every
   output wire~\cite{zahur2015two}. However, in \name, \AS does not have access
   to the circuit $F$ and the garbled inputs $(X,C)$. This makes it safe to
   disclose $\eir$ to \AS (\thref{thm:AS1}, \ref{thm:AS2}). Thus, \AS can
   compute $\wlabel^{w_1}_0, \ldots, \wlabel^{w_m}_0$ from the output labels
   (see \textsf{ASExec} in~\figref{fig:algo}) and check whether \TAP has
   returned forged labels for the output wires corresponding to $\ftwo$. The
   HMAC $\ench$ is used to ensure the authenticity of the first output wire
   corresponding to $\fone$, when it evaluates to $\false$. Because of this
   structure, \name achieves efficient authenticity verification with two hash
   values (\thref{thm:TAP2}). \change{This modification, combined with
   trusted generator, allows us to use efficient semi-honest
   implementations of GCs while achieving security against a
   malicious evaluator (\TAP).}

\noindent\change{ \textbf{(4)} We use a circuit id $j$ to synchronize between different parties (\TS,
  \TAP, \AS) so that they evaluate the correct circuit. Malicious \TAP can
  observe the circuit id (in plaintext) and can tamper with it. \name ensures that the \AS will always be able to catch a lying \TAP,
  and will never act on an incorrect circuit id $j$. (See the proof in~\apref{sec:sec_analysis}.)  Metadata leaked due to
  learning $j$ is outside the scope of this paper (Security Non-goals in \secref{sec:threat}). We discuss a potential
  solution in~\secref{sec:discussion}.}



\subsection{\change{Supporting \TAP-Specific Operations with Garbled Circuits}}\label{sec:boolfunc}
\change{While in theory any arbitrary function can be converted into Boolean circuits, and therefore can be computed using GCs, in practice they can be expensive. Via an analysis of existing real-world rules (\secref{sec:supported-functions}), we found that they involve well-defined and relatively simple Boolean and arithmetic operations --- these are well-studied and efficiently supported by existing GC libraries. }

\change{However, we also found that many rules use string operations, 
such as matching regular expressions and extracting or replacing substrings. The corresponding Boolean circuits of these operations, unless properly designed, will be inefficient to execute using GC~\cite{mohassel2012efficient}. 
\name computes these string operations by first translating regular expressions into deterministic finite automatons (DFA) and then applying a novel approach to convert DFA to Boolean circuits that can be efficiently evaluated using GC and can be easily extended for substring extraction and replacement. We next describe how \name utilizes this approach to perform regular expression matching. Details on substring extraction and replacement are given in~\appref{app:regex}. (Please refer to~\cite{chang1992regular} for details of how to convert a regular expression into a DFA.)
}

\paragraph{Input and output representations.} \change{
First, to avoid leaking the length of the string, every string field in the trigger data (and the action data) is padded to a fixed length bitstring. AS is responsible for removing the padding as necessary. 
The string is encoded into a fixed-length bitstring $\vecx = (x_1,\ldots,x_n)$ where $x_i\in\zo$ before feeding into the encoding function $\En$. 
Let the operation of the string be defined using the DFA $\dfa$, which is represented as a five-tuple,
$\dfa = (\states, \alphabet, \trans, \state_0, \finalstates)$, where $\states$
is the set of states, $\alphabet$ is the set of alphabets, $\state_0$ is the
initial state, and $\finalstates$ is the set of final states. The transition function
$\trans$ takes a state and an alphabet and returns the next state; therefore,
$\trans\Colon\states\times\alphabet\to\states$.  Since every string is a
bitstring, we have $\alphabet=\zo$. Let $\nstates = |\states|$ be the total number of states. Without loss of generality, we assume $\states = \mathbb{Z}_q  = \{1, \dots, q\}$.
}

\change{
Let $\transstring$ be the aggregated transition function that takes the entire string $\vecx$ as input and outputs the
final state of the DFA, 
\bnm \transstring(\vecx) = \trans(\ldots\trans(\trans(\state_0,x_1),x_2),\ldots, x_n).  \enm
If $\transstring(\vecx)\in\finalstates$, then
$\vecx$ is accepted by the DFA, which means that the string matches the regular
expression.
}

\paragraph{Converting DFAs into circuits.} \change{
The main goal is to convert the transition function $t = \trans(\state, x)$ into a Boolean circuit that uses as few AND and OR gates as possible, to take  advantage of the standard free XOR optimization~\cite{kolesnikov2008improved}.
}

\change{
Since both the states $\state$ and $t$ are integers between 1 and $q$, one can choose to represent each state using $\log_2 q$ bits and find the truth table for $\trans$. However, the resulting circuit would be hard to construct and minimize automatically.
Instead, 
we encode each state as a bit-vector of size $\nstates$ using one-hot encoding. 
We use $\State$ to denote the encoding of a state $\state \in \states$, and $\State^i$ represents the $i$-the bit of $\State$, where $\State^i = 1$ if $i=s$ and 0 otherwise. We can observe that when $S^i = 1$ and $x = 0$, $T^j=1$ if and only if $\trans(i, 0) = j$ holds; Similarly, when $S^i = 1$  and $x = 1$, $T^j=1$  if and only if $\trans(i, 1) = j$.
Therefore, the output of the DFA becomes %
\bnm \Transstring(\vecx) = \Trans(\ldots\Trans(\Trans(\State_0,
x_1),x_2),\ldots, x_n),  \enm
where $\Trans$ is the transition function that operates on the one-hot encoded states.
}

\change{
  To represent the transition function $\Trans$ as a Boolean circuit,  
we first define two sets for each state $s$, $P_0^\state$ and $P_1^\state$, where 
 $P_b^s = \{i\suchthat\trans(i,b) = j\}$ for $b\in\zo$.
It holds that $T^j=1$ if and only if either $x=1$ and $\exists i \in P_1^j, S^i=1$, or  $x=0$ and $\exists i \in P_0^j, S^i=1$. That is to say,  for  $1\le j\le q$,
\begin{align*}
\setlength{\abovedisplayskip}{4pt}%
\setlength{\belowdisplayskip}{4pt}%
\setlength{\abovedisplayshortskip}{6pt}%
\setlength{\belowdisplayshortskip}{6pt}
T^j  &= (x\wedge\bigvee_{i\in P^j_1}S^i) \vee (\neg x\wedge\bigvee_{i\in P^j_0}S^i) \\ 
&= (x\wedge\bigvee_{i\in P^j_1}S^i) \oplus (\neg x\wedge\bigvee_{i\in P^j_0}S^i). 
\end{align*}
Because only one of the $S^i$ will be 1 at any time, therefore the inner OR gates can also be replaced with XOR:
\bnm T^j  = (x\wedge\bigoplus_{i\in P^j_1}S^i) \oplus (\neg x\wedge\bigoplus_{i\in P^j_0}S^i).\enm 
Note the above expression can be further simplified using the Boolean algebra property 
$(x\wedge a) \oplus (\neg x\wedge b)= ((a \oplus b) \wedge x) \oplus a$.
Therefore, each bit in $T$ requires at most one AND gate to compute. To run $\dfa$ over a string of length $n$, we need to apply transition function ($\Trans$) $n$ times, and thus the resulting circuit contains at most $n\nstates$~AND gates.
Finally, to check if the final state is accepted by $\dfa$, simply computing $\bigoplus\limits_{j\in\finalstates}S_n^j$ is sufficient. 
}

\change{
We can observe that the size of the entire garbled circuit is $O(n\nstates\keylen)$, on par with the communication cost of the state-of-the-art non-GC based customized approach~\cite{mohassel2012efficient}. However, being purely circuit-based, our approach allows functional conjugation with other operations and retains the same security properties of standard GC. 
We describe more details on how to
extend this approach to perform substring extraction and replacement
in~\appref{app:regex}.
}

\paragraph{Supported functions.} 
By incorporating the above techniques, we can use garbled circuit to efficiently compute 
\change{common arithmetic operations, string operations,  
  and dictionary lookup, which cover}
 all but three functions listed in~\figref{fig:functions}. We sketch
the implementation details for each supported function
in~\appref{app:func-impl}. Based on our analysis in~\secref{sec:supported-functions}, this set of operations enables \name to
support $93.4\%$ of the function-dependent rules published on Zapier and \emph{all} of the~500 most popular rules on IFTTT.

It is possible to convert the remaining three unsupported functions (\texttt{format}, \texttt{strip\_html}, and \texttt{html2markdown}) to Boolean
circuits, as well, but the resulting circuits will be very large (for example, we
need to build a full-blown parser to find HTML tags) and inefficient to evaluate. These functions
are only used for formatting and do not require any sensitive user input. Thus, it is safe to run them on \AS or \TS directly with minor modifications to their APIs.



\section{Evaluation of \stap}
\label{sec:eval}
\newcommand{\ptap}{\textsf{PlainTAP}\xspace}

We prototyped \stap and showed that it is competitive in
performance with TAPs that do not provide any data
privacy.
We implemented the garbled circuit protocols described in~\secref{sec:design}
using EMP toolkit~\cite{emp-toolkit}, a C++ library for multi-party
computation. We build on EMP toolkit's semi-honest 2PC protocol. We
use state-of-the-art optimizations (including free
XOR~\cite{kolesnikov2008improved} and half gates~\cite{zahur2015two}) for
improving efficiency and bandwidth. The security parameter is 
$\keylen = 128$.  For other cryptographic operations we use
Cryptography.io~\cite{crypto:io}.  We use SHAKE-128 (a member of SHA-3 family~\cite{sha3standard}) as a cryptographic hash
function, and AES in CBC mode with HMAC using SHA-256 as a semantically secure,
non-malleable, robust symmetric-key encryption scheme.  To convert regular
expressions into DFAs we use the library dk.brics.automaton~\cite{dfa-lib}. For all experiments, we used \textsf{n1}-standard instances in Google Cloud Platform configured with 2 vCPUs, 7.5~GB memory, and 1~Gbps network
connection.







\subsection{Performance of Basic Operations}
\label{sec:micro}

\name supports Boolean, (integer) arithmetic, and string operations (which is
sufficient to run most of the rules in Zapier and IFTTT).  To evaluate the
performance of these basic operations, we picked a set of representative
operations from~\figref{fig:functions}. For Boolean, we chose the AND operation
since our circuits only contain AND and XOR gates, and the XOR gate can be
computed without any encryption costs~\cite{kolesnikov2008improved}. For numeric
data, we selected comparison and multiplication between two 32-bit integers.
For string operations, we divided them into two categories: operations that need
regular expressions (\texttt{contain}, \texttt{replace}, \texttt{split}, and
\texttt{extract\_phone}) and those that do not (\texttt{lookup} and
\texttt{==}). We set the input \texttt{x} as a 100-character (800 bit) string,
except for \texttt{lookup}, where we set \texttt{x} to a 10-character string. In
the function {\tt m.lookup(x)}, we set {\tt m} to be a key-value store with 10
entries, where each key and each value is 10-characters long. For
\texttt{x.replace(s, "")} and {\tt x.contain(s)}, we set the \texttt{s} to a
4-character string. For {\tt x.split(d, 0)}, we set {\tt d} to be a single
character.

While measuring the costs for above basic garbled circuit operations, we do not consider the overhead of other components like payload encryption, as they are independent of the operation. \figref{fig:micro} shows the time required for each operation.

\begin{figure}[t]
\centering\tabfontsize
\begin{tabular} {p{1em}p{8em}rrrrrr} 
 \toprule
  \multicolumn{2}{c}{\multirow{2}{*}{\textbf{Operation}}} &  \multicolumn{4}{c}{\textbf{Computation time} (ms)} & \textbf{GC size} &  \textbf{\# DFA} \\ 
 & & {Client} & {TS} & {TAP} & {AS} & (KB) &\textbf{states} \\
 \midrule
Bool & \texttt{x \& y} & 4.0 & 3.7 & 3.7 & 3.9 & 0.03 & -- \\
\cmidrule{1-8}
\multirow{2}{*}{Num} & \texttt{x > n} & 4.0 & 3.9 & 3.8 & 3.8 & 0.96 & --\\
& \texttt{x * n} & 4.0 & 3.7 & 4.0 & 3.7 & 31 & --\\
\cmidrule{1-8}
\multirow{6}{*}{Str} & \texttt{x == t} & 4.0 & 3.7 & 4.0 &  3.8 & 25 & --\\
 & \texttt{m.lookup(x)} & 4.2 & 3.6 & 4.1 & 3.8 &  31& -- \\
 & \texttt{x.split(d,0)} & 5.7 & 3.7 & 5.3 & 4.1 & 78 & 16\\
& \texttt{x.contain(s)} & 7.8 & 3.9 & 7.4 & 3.9 & 123 & 47 \\
& \texttt{x.replace(s,"")} & 10.7 & 3.8 & 10.5 & 4.6 & 278 & 40 \\
&  \texttt{x.extract\_phone()} & 24.7 & 3.6 & 25.5 & 4.1 & 2191 & 108 \\
\bottomrule
\end{tabular}
\caption[]{Execution time of different basic operations at the
  client (TC), the trigger service (TS), the action service (AS), and the TAP. We record the
  size of the garbled circuit sent from TC to TAP and the number of states in the DFA if applicable.
}
\label{fig:micro}
\end{figure}

The circuit generation (at \TC) and circuit evaluation (at TAP) take roughly the same amount of time for each operation, which is expected because they require roughly similar operations. Most of the Boolean, arithmetic,
and some string operations (such as string equality or lookup) execute in less than 4 ms on the \TAP.  Complex string operations are also fast (takes less
than 25 ms) under some reasonably sized inputs.  TS and AS can encode/decode inputs in less than 5~ms.


We record the size of the garbled circuit ($|F|$) for each operation in
the second-to-last column of~\figref{fig:micro}.  The garbled circuit
$F$ needs to be periodically transferred from the client to TAP and the size of
the circuit changes significantly for different operations.  Although for Boolean
AND the circuit is only 31 bytes, the circuit size for a complex regular expression
extraction, which is one of the most expensive operations we found, is quite large
(2.2~MB).  The size of garbled circuit
increases with the number of states in the DFA and the length of the input
string. 
The string replacement circuits (\texttt{replace}) are larger, about 2.25x, than their equivalent
matching circuits (\texttt{contain}), even though the required DFA is larger for the latter operation. The lookup circuit is
small ($31$~KB). 




\subsection{Performance of Running Complete Rules}
\label{sec:macro}
Next, we measure the performance of \stap on real-world rules.
We first picked ten rules from the combined IFTTT and Zapier dataset we
collected in~\secref{sec:supported-functions}. These rules handle sensitive data of different sizes and cover a wide variety of operations (as noted~\figref{fig:functions}). 
We list the rules with simple descriptions 
in~\figref{fig:rules-exp}. The first eight rules (R1-R8) involve frequently used functions, while the last two rules represent two rare but extreme
scenarios. R9 requires a rarely-used \texttt{extract\_phone} function, which appears only three times in our dataset and requires a complex regular
expression to be evaluated over a long text, thus making it the most expensive
rule to compute in our
dataset. 
R10 is connected to a trigger that might have a large payload (videos), so
its performance is more dependent on network bandwidth and latency.


For comparison, we built a skeleton version of each service following the
current TAP model, where only plaintext data is exchanged and computed, as a
baseline. We refer to this as $\ptap$.
We used Python library Flask for the cloud component of TAP, as well as
two RESTful servers that mimic the APIs provided in current trigger and action
services. \change{Two US-west
instances hosting TS and AS, and two US-central instances for hosting TAP and the (simulated) TC. The network latency between US-west
and US-central is $39$~ms.}


\paragraph{Latency.} The end-to-end execution latency measures the time between a
trigger event (trigger data and payload are available to \TS) and \AS receiving plaintext output (\figref{fig:throughput}). The latency, except R9 and R10, is below 260~ms.  When compared to \ptap (\figref{fig:throughput}, top),
the execution latency for \stap is $55\%$ more on average. The majority of the latency
overhead is due to the higher amount of data transfer in \name between \TS and
\TAP ($27$-$51$~KB) and between \TAP and \AS ($3$-$32$~KB), which is nearly $128$x more
than what it would require in \ptap. We show the data transfer in the last two
columns in~\figref{fig:rules-exp}.  Given that \TS, \AS, and \TAP are cloud-based
services with high-bandwidth network links, the increase in data
usage is reasonable.
\change{In addition, we list the time spent by \TC to generate and upload a single circuit (as the red bar in \figref{fig:throughput}, top).
  \TC needs less than $12$~ms to generate and transfer one circuit for most rules (except for R9, in which case it takes $172$~ms).  
This metric  represents the setup time for a new rule before it can be executed.
In practice, \TC can generate and upload circuits in bulk periodically at its convenience. 
}

%

\pgfplotstableread{
rule
1
2
3
4
5
6
7
8
9
}\datatable
\pgfplotstableread[col sep=comma]{data/latency_updated.txt}\latencydata
\begin{figure}[t]
  \centering

  \begin{tikzpicture}
    \begin{groupplot}[
      group style = {group size = 2 by 1, , horizontal sep=0.5cm},
      ]
      \nextgroupplot[xmin = 0, xmax = 27, 
      ybar stacked,
      legend style={
        legend columns=3,
        at={(xticklabel cs:0.55)},
        anchor=north,
        draw=none,
        font=\footnotesize,
        yshift=-2ex,
    },  
    xtick=data,
    bar width=1.5mm,
    ymin=0,
    ymax=0.6,
    width=7.5cm,
    height=4cm,
    ytick distance={0.2},
    minor tick num=3,
    xticklabels from table={\datatable}{rule},
    y tick label style={font=\footnotesize},
    x tick label style={font=\footnotesize, xshift=+0.5ex},
    tick label style={font=\footnotesize},
    label style={font=\footnotesize},
    xlabel={rule \#},
    ylabel={Latency (sec)},
    area legend]
    \addplot+ [pattern=north east lines, pattern color=blue] table[x=idx,y=1] {\latencydata};
    \addlegendentry[]{\name execution~~~};
    \addplot+ [pattern=north east lines, pattern color=red, dotted] table[x=idx,y=3] {\latencydata};
    \addlegendentry[]{\name setup~~~};
    \addplot+ [brown, pattern=north west lines, pattern color=brown] table[x=idx,y=2] {\latencydata};
    \addlegendentry{\ptap execution};

    \nextgroupplot[xmin = 27, xmax=30,
    ybar stacked,
    bar width=1.5mm,
    ymin=0, ymax=3,
    width=2.25cm,
    height=4cm,
    ytick distance={1},
    minor y tick num=3,
    xtick={28,29},
    xticklabels={10},
    x tick label style={font=\footnotesize, xshift=+0.5ex},
    tick label style={font=\footnotesize},
    label style={font=\footnotesize},
    ]
    \addplot+ [pattern=north east lines, pattern color=blue] table[x=idx,y=1] {\latencydata};
    \addplot+ [pattern=north east lines, pattern color=red, dotted] table[x=idx,y=3] {\latencydata};
    \addplot+ [brown, pattern=north west lines, pattern color=brown] table[x=idx,y=2] {\latencydata};
    \legend{}
  \end{groupplot}

\end{tikzpicture}
\begin{tikzpicture}
  \pgfplotstableread[col sep = comma]{data/throughput.txt} \throughputdata

\begin{groupplot}[
     group style = {group size = 2 by 1, , horizontal sep=0.5cm},
    ]
    \nextgroupplot[xmin = 0, xmax = 27, 
    ybar stacked,
     legend style={ legend columns=2,
                    at={(xticklabel cs:0.5)},
                    anchor=north,
                    draw=none},  
     xtick=data,
     bar width=1.5mm,
     ymin=0,
     ymax=230,
     width=7.5cm,
     height=4cm,
     xticklabels from table={\datatable}{rule},
     y tick label style={font=\footnotesize},
     x tick label style={font=\footnotesize, xshift=+0.5ex},
     tick label style={font=\footnotesize},
     legend style={font=\footnotesize,yshift=-2ex},
     label style={font=\footnotesize},
     xlabel={rule \#},
     ylabel={Throughput (exec/sec)},
     area legend]
      \addplot+ [pattern=north east lines, pattern color=blue, x tick label style={xshift=-0.3cm}] table[x=idx,y=Throughput] {\throughputdata};
        \addlegendentry[]{\name~~~~};
        \addplot+ [brown, pattern=north west lines, pattern color=brown, x tick label style={xshift=-0.3cm}] table[x=idx,y=Throughput IFTTT] {\throughputdata};
        \addlegendentry{\ptap};
        
    \nextgroupplot[xmin = 27, xmax=30,
    ybar stacked,
     bar width=1.5mm,
     ymin=0, ymax=8,
     width=2.25cm,
     height=4cm,
    xtick={28,29},
     xticklabels={10},
     x tick label style={font=\footnotesize, xshift=+0.5ex},
     tick label style={font=\footnotesize},
     label style={font=\footnotesize},
     area legend]
      \addplot+ [pattern=north east lines, pattern color=blue, x tick label style={xshift=-0.3cm}] table[x=idx,y=Throughput] {\throughputdata};
        \addplot+ [brown, pattern=north west lines, pattern color=brown, x tick label style={xshift=-0.3cm}] table[x=idx,y=Throughput IFTTT] {\throughputdata};
  \end{groupplot}
\end{tikzpicture}

\caption{Latency (top) and throughput (bottom) for running each of the rules (X-axis) in \stap and \ptap.}
\label{fig:throughput}
\end{figure}
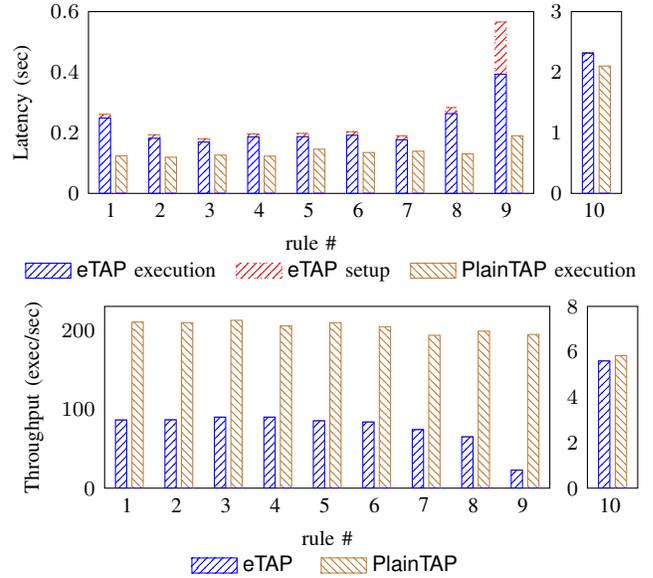

\paragraph{Throughput.} We measured the throughput as the maximum number of
executions per second by \stap. We used Apache Bench~\cite{apache-ab} to compute the
throughput, which simulates sending concurrent trigger messages to \stap. We
pre-computed the trigger labels to eliminate the bottleneck
on TS. We gradually increased the concurrency level until the throughput
saturated. We reported the maximum throughput of \stap and $\ptap$
in~\figref{fig:throughput} (bottom). \stap is capable of executing 65-90 rules of type
R1-R8 per second on a single server. 
Compared to \ptap, for all but one (R9) rule, \stap provides around $41\%$
throughput of \ptap. In the worst case, when executing R9, \stap's throughput reduces to $11\%$ of \ptap.

\subsection{Large-scale Evaluation}\label{subsec:eval-large}
\change{
To better characterize the performance of \name under realistic workloads, we performed a large-scale evaluation where we randomly sampled $100$ rules from our combined IFTTT and Zapier dataset. Out of the $100$ sampled rules, $55$ require computations on the trigger data. For rules with no computations, we simply treated the trigger data as payload and encrypted them inside $ct$. 
}

\paragraph{Computation overhead on TC.}
\change{In \name, the trusted client \TC has to
periodically generate and distribute the garbled circuits $F$, associated garbled constants $C$, and encrypted decoding blobs $\encd$ to
\TAP. For simplicity, we will use the term garbled circuit to denote the set $(F,C, \encd)$.
On average it takes $4.1$~ms to generate one garbled circuit.  Based on
prior work~\cite{cobb2020risky}, we assume that an average user has $26$ rules
installed and that each rule will be executed once every $15$~minutes, which is
the default interval used by IFTTT to contact its trigger services~\cite{mi2017empirical}. Therefore, we estimate that the \TC of an average user needs to spend $10.2$ seconds per day to
generate $2,496$ circuits.  Since the average circuit size is $25.3$~KB, the
estimated amount of data that \TC has to send to TAP per day is $61.7$~MB, which
is less than the data required to back up $25$ high-res photos or a 1-minute HD video
(1920x1080 px @30 fps) to a cloud service~\cite{cnet-apple-photo, lifewire-apple}, a common task
executed daily by modern smartphones.}

\paragraph{Storage overhead.}
\change{
\TAP needs to store all circuits uploaded by \TC until they are executed.
Based on the dataset in~\secref{sec:supported-functions},
there are $12.4$~million rules (counted by number of installations) running in IFTTT that are connected to private triggers. If we assume, conservatively, that all of them require computations, that each rule will be executed once every 15 minutes, and that each rule on average requires $25.3$~KB of storage per execution based on the sampled rule set, the total storage overhead for using \name would be $28$~TB. Given that cloud storage is inexpensive~\cite{ebs-pricing}, the
overhead is manageable.
}
\change{  \name introduces little storage overhead to \AS, \TS, and
  \TC. \TS and \AS only need to store a $16$-byte key ($\sects$ and $\secas$) and the current circuit id $j$
  ($4$~bytes) for each user. 
  \TC needs
  to store the circuit id $j$ and the keys for each service connected to the
  user's installed rules, since it can delete the circuits it generated after
  uploading them to TAP.  }

\paragraph{Latency and throughput.}
\change{We first measured the end-to-end latency of running each rule individually and computed the average. The average latency of \name
  is $139$~ms, which is similar to $\ptap$ ($110$~ms). The increase in latency should be tolerable, considering the delays in current trigger-action systems are usually 1 to 2 minutes \cite{mi2017empirical}. Then, we issued concurrent requests to trigger every rule at the same time and recorded the maximum throughput. The throughput of \name is $96$~requests per second (RPS),
  which is $45\%$ of the throughput of \ptap ($211$~RPS). Overall, we have shown that \stap can run real rules with a modest performance impact. 
}

\section{Related work}\label{sec:related}
A few studies have investigated the security issues in IFTTT-like systems. Most closely related is the work of Fernandes et al.~\cite{fernandes2018decentralized} where they first introduce the compromised TAP model, and then built DTAP, a system to prevent the misuse of stolen OAuth tokens. They focus only on the integrity problem. By contrast, our work subsumes DTAP by providing confidentiality to the private trigger data passing through TAPs and  adding authenticity of trigger-compute-action rule execution.

\change{Chiang et al.~\cite{otap} recently propose Obfuscated TAP that handles metadata attacks. They propose techniques to hide trigger data arrival patterns and the types of trigger and action services from the untrusted TAP. Their work also performs end-to-end encryption of trigger data but cannot support computations. In contrast, \name focuses on protecting sensitive trigger data while allowing computation --- a common use-case in real-world rules (e.g., filter codes in IFTTT). }


Bastys et al.~\cite{DBLP:conf/ccs/BastysBS18} classify the sensitivity of IFTTT's trigger and action services and show that $30\%$ of IFTTT's apps may violate privacy by exfiltrating private information to a third-party. Xu et al. ~\cite{DBLP:journals/access/XuZZCDG19} analyze how much private data can be harvested by TAPs. They demonstrate that IFTTT has access to more data than necessary. For example, IFTTT monitors devices even if they do not trigger actions. This motivates our work in protecting all information from a malicious TAP. 

A popular line of work investigates the semantics of rules and how they violate security policies or interfere with each other. Surbatovich et al.~\cite{DBLP:conf/www/SurbatovichABDJ17} present an empirical study of IFTTT apps and categorize the apps with respect to potential security and integrity violations. Wang et al.~\cite{iruler} design iRuler that uses SMT techniques to discover inter-rule vulnerabilities. This work is orthogonal to ours as it deals with rule semantics and the TAP is considered trusted. By contrast, our work protects trigger data from a malicious TAP.

\paragraph{Cryptographic techniques for secure computation.} There is a large
body of work on privacy-preserving outsourced computation. Garbled circuit is a particularly popular approach
~\cite{DBLP:conf/icalp/KolesnikovS08,DBLP:conf/uss/KreuterSS12,DBLP:conf/uss/HuangEKM11,carter2016secure,songhori2015tinygarble}. However,
most practical approaches tend to be
application-dependent~\cite{DBLP:journals/jcs/KolesnikovS013,bringer2012faster}. Since
our setting differs from a generic multi-party setting (as discussed in Section \ref{sec:design}, we needed to
develop a customized protocol).

\change{
For evaluation of string operations in a secure two-party computation
setting, Mohassel et al.~\cite{mohassel2012efficient} introduced
Obliv-DFA, a custom non-GC based protocol that
only supports regular expression matching. Extending Obliv-DFA to substring extraction and replacements would not be possible without drastically changing the protocol and incurring significant overhead. 
\name, on the other hand, supports string operations through a novel and efficient purely circuit-based approach. This allows functional composition with other operations such as substring extraction/replacements and simple transfer of security properties of GC.
}


\section{Discussion and Limitations}
\label{sec:discussion}

\paragraph{Security against metadata leakage.} Some rules reveal
sensitive information just because they are executed. For example, consider the
rule: ``IF I leave home, THEN turn off the WiFi.'' TS sends a message
to AP only when the user leaves the home. In our threat model, TAP
knows the rule semantics. Therefore, when TAP observes a message from this particular TS, it will learn that the user has left the home. 
\change{Such metadata leakage from side-channels is hard to prevent cryptographically.
Recent work~\cite{DBLP:journals/access/XuZZCDG19,otap} has applied \emph{cover traffic} to protect time-sensitive information in trigger-action rules by hiding the real trigger events among fake-but-identical ones. We discuss a simple modification to \name that uses cover traffic without requiring \TC to generate new (fake) circuits.}

\change{
\TC generates a set of circuits and transmits them to \TAP, as before. Let $J$ denote circuit indices in this set. TS also internally keeps track of the set of circuit indices $J'$ that have been used with \emph{real} data. 
To send real data, \TS  picks random $j\in J\setminus J'$, updates $J'\gets J'\cup\{j\}$, and continues as before (\secref{sec:design}). To send fake data, it picks random $j \in J$, and then sets the garbled trigger data to random bits. \TAP executes the chosen circuit ID as before and sends output to \AS. When fake data is evaluated on a circuit, the decryption at the \AS will fail with very high probability and thus, it will ignore the message.
We present an example to illustrate this approach. Assume \TC generates two circuits with ids $J = \{j_1, j_2\}$ and \TS has to send five events $e_1, e_2, \ldots, e_5$, among which only $e_2$ and $e_3$ are real. 
Following the scheme above, it transmits the following sequence of circuit ids to \TAP: $j_1, j_1, j_2, j_1, j_2$. We see that $j_1$ is used multiple times: first for $e_1$, then for $e_2$, and finally for $e_4$. \TAP will notice that $j_1$ circuit was executed thrice, but it cannot distinguish which of these executions was on real data. 
}

\change{This approach is secure due to two reasons: (1) \TAP cannot distinguish between executions on real or fake data due to the garbling procedure 
  we use in \name~\cite{zahur2015two} (see~\propref{thm:input} for a proof-sketch); and (2) \TAP cannot learn anything from circuit usage statistics because of how \TS selects $j$. In addition, circuits can be executed multiple times but \emph{at most only one} of them will be on real data. 
Such re-evaluation of circuits on random data does not affect GC security properties \cite{bellare2012foundations}. }

\paragraph{Integrating with existing Trigger-Action Systems.}
We contribute a clean-slate redesign for
trigger-action platforms providing data confidentiality from the ground up. As
such, it is not immediately backward compatible. \change{However, \name's design attempts to minimize these required changes as follows:}

\change{First, we create a new $\TC$, a mobile app that users must install on
  their phones to interact with \name. The app mimics
  the user experience that trigger-action platforms like IFTTT or Zapier
  currently offer. For example, the user clicks on buttons in a wizard-style
  user interface to program a rule. \TC transparently generates keys and GCs in the background and shares them with \TS, \TAP, and \AS
  (accordingly) --- the user does not have to take any additional action.}

Second, the existing \TS and \AS need to adapt to \name
  protocol. Specifically, both need to communicate with \TC to receive keys
  ($\sects, \secas$). Additionally, \TS has to send encoded labels to \TAP
  instead of plaintext trigger data, and \AS has to run the decoding function on
  circuit output (\figref{fig:algo}). We have built a 
  library that
  trigger/action services can use to upgrade their APIs to perform the above
  operations. 

\change{Third, \TAP has to evaluate GCs. It also has to cache
  circuits it receives from \TC. We observe that \TAP is already setup to
  perform these tasks --- executing code at large scale and managing
  user-specific data. Although this incurs a resource cost, we believe that it
  is acceptable given the strong confidentiality and integrity guarantees our
  work provides.}

\paragraph{Rule semantics.} A malicious TAP can learn about a user's automation
patterns using its knowledge of rule semantics. Although we encrypt the
trigger data, \TAP can still observe the source endpoint of the trigger data and the destination of the encrypted result. As
future work, we envision using results from anonymity networks like
Tor~\cite{tor} to hide the sources (trigger service) and destinations (action
service) of messages. 

\paragraph{Circuit id synchronization.} \name requires \TC and \TS to synchronize on the circuit id $j$. 
\TAP in \name cannot execute a rule if the $j$ specified by  \TS is not present
in its database of GCs sent by \TC. This can happen, for example, if  \TC
fails to generate circuits for a certain day due to technical glitches, but
 \TS continues to generate trigger data.  We do not want \TS to support additional
APIs to inform  \TC about its current circuit id $j$. Instead, we can rely on
  \TAP to provide this information. 
 \TS attaches an encrypted (using $\sects$) blob containing the circuit id
and the timestamp to $\TAP$ along with other data during rule execution. 
$\TAP$ forwards that blob to \TC on request from \TC. Thus  \TC can learn the current value of  
$j$ and can detect if  \TAP sends a stale message.  



\paragraph{Loss of the trusted client (TC).} TC in our setting is the ``root'' of trust for generating garbled circuits.  TC can be an app
running on user's personal mobile device.  However, the app has to store a
number of important states necessary for continued execution of a rule, such as
$\sects, \secas, \mbox{OAuth tokens}, j, f, c$, etc. Therefore, the states on the trusted
client must be preserved in case the device is lost. We can use standard
cloud-based solutions to back up the states. For example, the states can be
encrypted under a user's password and backed up in a cloud drive.
The client can recover the states and continue to operate on a new device once the
user connects their cloud drive accounts.



\paragraph{Circuit usage feedback.} Different rules execute at varying rates. \TAP can monitor rule execution frequency to make predictions about future circuit usage and optimize the number of circuit generations and transmissions. \TAP can lie about these statistics; however, it does not affect on the security of \name. We leave its implementation to future work.


\section*{Acknowledgements}

We thank the anonymous reviewers for comments that helped improve
the paper.  This work was supported in part by the University of
Wisconsin-Madison Office of the Vice Chancellor for Research and Graduate
Education with funding from the Wisconsin Alumni Research Foundation.
This work was also partially supported by the Swedish Foundation for Strategic
Research (SSF) and the Swedish Research Council (VR).


\bibliographystyle{abbrv}
\bibliography{bib}

\appendix
\subsection{Security Analysis of \stap}
\label{sec:sec_analysis}

\begin{figure*}[t]
  \tabfontsize\centering
  \hfpagesss{0.315}{0.31}{0.31}{
    \underline{$\obliv$:}\\[2pt]
    $\left(f, (\dt^0, \const^0, \dtp^0), (\dt^1, \const^1, \dtp^1)\right)\getsr \adv $\\
    $\mbox{Pick } j$ ;\; $\sects\getsr \{0,1\}^\kappa$ ;\; $\secas\getsr \{0,1\}^\kappa$\\
    $b\getsr\{0,1\}$\\
    $j,F,C,\encd \getsr \cktgarbling\left((f,c^b), (\sects,\secas,j)\right)$\\
    $j,X,ct \getsr \tsexec\left((x^b,\payload^b),(\sects, j)
    \right)$\\
    $b'\getsr \adv(j,X,ct, F, C, \encd)$\\
    $\return b=b' $
   }{ 
    \underline{$\auth$:}\\[2pt]
    $\left(f, (\dt, \const, \dtp)\right)\getsr \adv $\\
    $\mbox{Pick } j$ ;\; $\sects\getsr \{0,1\}^\kappa$ ;\; $\secas\getsr \{0,1\}^\kappa$\\
    $j,F,C,\encd \getsr \cktgarbling\left((f,c),(j,\sects,\secas)\right)$\\
    $j,X,ct \getsr \tsexec\left((x,\payload),(\sects,j) \right)$\\ 
    $j', \Ytilde, ct', \encd'\getsr  \adv(j,X,ct, F, C, \encd)$\\
    $y' \gets \asexec\left((j',\Ytilde,ct',\encd'),\secas\right)$\\
    $\return (j', \Ytilde, ct', \encd') \ne (j, F(X), ct, \encd)$\\
    $\mytab\mytab \wedge y'\ne \bot $
  }    {
    \underline{$\privone$:}\\[2pt]
    $\left(f, (\dt^0, \const^0, \dtp^0), (\dt^1, \const^1, \dtp^1)\right)\getsr \adv $\\
    If $f(\dt^0, \const^0) \ne f(\dt^1, \const^1)$ then \ Return $\bot$\\
    $\mbox{Pick } j$ ;\; $\sects\getsr \{0,1\}^\kappa$ ;\;$\secas\getsr \{0,1\}^\kappa$\\
    $b\getsr\zo$\\
    $j,F,C,\encd \getsr \cktgarbling\left((f,c^b), (\sects,\secas,j)\right)$\\
    $j, X, ct \getsr \tsexec\left((x^b,\payload^b), (\sects, j)\right)$\\
    $j, Y, ct, \encd \gets \tapexec\left((j, X, ct), (F, C, \encd)\right)$\\
    $b'\getsr \adv(j,Y,ct,\encd)$\\
    $\return b=b' $
  }
  \caption{Security games for \name.}
  \label{fig:games}
\end{figure*}

In this section, we show that \stap meets the security goals outlined
in~\secref{sec:threat} by providing concrete security definitions and proofs.
We assume the adversaries are probabilistic polynomial time (ppt) ---
they 
run in time polynomial in security parameter $\keylen$. The garbled circuit protocol $\G$ used in \name provides \emph{output privacy}, \emph{message obliviousness}, and
\emph{execution authenticity}. The encryption scheme $\ske$ is
\textsf{IND-CCA} secure. We model the hash function $\H$ as a random
oracle~\cite{RO}. Let $\negl(\cdot)$ to be a negligible function. 

We prove the security of each component of \name, namely TAP, TS, and AS, separately. The
security games are defined in~\figref{fig:games}. 



\paragraph{\em Security against malicious TAP.}
Following our threat model, we assume the TAP is compromised and
\emph{malicious}. 
The security definitions we expect from \stap are as follows.

\paragraph{Obliviousness.}
We define the obliviousness property of $\stap$ by the security game $\obliv$ as
shown in~\figref{fig:games}.  Informally, $\advA$ despite arbitrarily deviating
from the protocol should not know anything about the user-provided constants
$\const$, the trigger data $\dt, \dtp$, and the output of the function
$y \gets f(x)$.

\begin{theorem}[TAP Obliviousness]
  For any ppt adversary $\advA$, the probability that $\advA$ wins the $\obliv$
  game is negligible. \bnm \Prob{\obliv = 1}\leq 1/2+\negl(\keylen)\,,\enm
  \label{thm:TAP1}
\end{theorem}\vspace{-0.1in}
\begin{proof} The proof of this theorem follows directly
from the \emph{message obliviousness} security guarantee of garbled circuits
$\G$~\cite{zahur2015two} and the semantic security of the encryption scheme
$\ske$. As such, the attacker learns nothing about $(\dt, \payload, c)$ from
$(X, C, ct)$. First, note that the game $\obliv$ is equivalent to the game $obv.sim_{\mathcal{S}}$ \cite{bellare2012foundations} in \cite{zahur2015two}. Now, consider the simulator $\mathcal{S}$ as presented in Fig. 3 in \cite{zahur2015two}. In our setting, $\mathcal{S}$  is used by \TC and \TS to generate $(\hat{F},\hat{X},\hat{C})$ which is then used for the rest of the computation. Hence the obliviousness of ($x,c$) follows directly from the corresponding proof (game $obv.sim_{\mathcal{S}}$) presented in \cite{zahur2015two} assuming the random oracle model for $\H$ \cite{RO}.  The indistinguishability of $ct^b$ follows trivially from the semantic security guarantee of the encryption scheme, thereby concluding our proof. \end{proof}

We achieve security against a malicious \TAP even with a GC implementation
for the semi-honest model. Recall that the ``generators'' --- the trusted client
(TC) and the trigger service (TS) --- in \name are at least semi-honest.  Hence, 
a valid garbled circuit for the correct function $f$ is always generated (as
\TC is trusted), and all inputs are correctly encoded (since \TS is semi-honest
and the ``evaluators'' \TAP and \AS have no input).  Thus, the only way a
malicious \TAP can compromise the security of \name is by forging an inauthentic
output label or by replaying, delaying, or dropping a message.  We
discuss \name's resilience to such attacks next.

\paragraph{Authenticity.} The security guarantee \textit{authenticity} ensures
that no ppt adversary can create a garbled output $\Ytilde\ne Y$ such that \AS
acts on $\Ytilde$ (that is to say $\asexec$ outputs anything but $\bot$ or $\false$).  The formal definition is given by the security game $\auth$
as shown in~\figref{fig:games}.

\begin{theorem}[TAP authenticity]
  For any ppt adversary $\advA$, the probability that $\advA$ wins the 
  game $\auth$ is negligible, \bnm \Prob{\auth = 1}\leq \negl(\keylen)\,. \enm 
  \label{thm:TAP2}
\end{theorem}\vspace{-0.1in}
\begin{proof} The proof follows from the non-malleability 
guarantee (IND-CCA) of the encryption scheme $\ske$, execution authenticity of
$\G$~\cite{zahur2015two}, and the collision resistance of the hash function
$\H$. 
For the rest of the proof, consider the simulator $\mathcal{S}$ in \cite{zahur2015two} which additionally returns $h=\H(\wlabel^{w_1}_0 \concat \ldots \concat \wlabel^{w_m}_0), e_r$ and $\wlabel^{w_0}_0$. \TC uses this additional information to generate the decoding blob $\encd$. Similarly, the function in $\De$ is changed to that of \textsf{ASExec}.
\vspace{0.03in}

\noindent\emph{Case I - Authenticity of $y_1=\fone(x,c)$.}\\
Note that $\encs$ is encrypted under a key derived from $\secas$ and $\wlabel^{w_0}_0$. Hence, from the semantic security of the encryption scheme, \TAP does not have access to $e_r$ since it does not know $\secas$ by design. Thus, 
in case $y_1=\false$, \TAP has access only to the $\false$ label $\wlabel^{w_0}_0$ and thereby cannot cheat \AS. On the other hand, if $y_1=\true$, \TAP can return some garbage value $L'$ such that $\skedec(L'\xor\secas,\encd)=\bot$. However, \AS can detect this with the help of the $\textsf{HMAC}$. Moreover, 
\TAP cannot send any of the hitherto unseen \textsf{HMAC}s  because it cannot obtain the output labels without access to  the corresponding $X$  (\TS's trigger data). 
\vspace{0.03in}

\noindent\emph{Case II - Authenticity of $y_2=\ftwo(x,c)$. } \\  
From the collision resistance of $\H$, the only way \TAP can cheat is by generating a label $\wlabel^{w_i}_{1-y_2[i-1]}$ for some wire $i \in [1,m]$ where $y_2[i]$ denotes the $i$-th bit of $y_2$. However,  as discussed above, \TAP cannot compute any other label other than the one obtained from $\Ev(F,X,C)$.

The rest of the proof follows an identical sequence of hybrids as the proof of Theorem 1 in \cite{zahur2015two} assuming the random oracle model for $\H$. \end{proof}

\paragraph{Protection from altering the timing of rule execution.}  An adversary
cannot forge a message that the \AS will accept due to the strong authenticity
guarantee of \stap protocol. However, it can alter the execution time of a rule
by deliberately dropping, delaying, or replaying messages. \TAP can successfully drop a message without being detected by \AS. However, this would fall under the denial-of-service attack which is beyond \name's scope (\secref{sec:threat}). \name also protects against replayed or delayed messages. 

Every message from \TS is timestamped as they are sent which $\AS$ can check before performing any action. Therefore,
\AS will reject a message --- outputting $\bot$ ---  if the received message is delayed more than $\delay$ seconds (a parameter set by \AS) since
the time it was sent from \TS. (See the function \textsf{ASExec}
in~\figref{fig:algo}.) 
We
acknowledge that the TAP can replay any message for which $\fone(x, c) = \false$
without getting detected by the AS.

Nevertheless, this does not lead to any undesirable outcome in practice because in this case \AS performs no action. Note that the above attack (replay of $\false$ labels) could have been prevented by keeping track of the last seen circuit id of each rule at \AS. However, maintaining such state information would violate the RESTfulness of \AS.

\paragraph{Tampering with circuit id $j$.}
\change{The malicious TAP can modify the circuit id $j$ --- a unique identifier
  given to every instance of a garbled circuit for synchronization between TAP,
  TS, and AS --- in whatever way they want to. But \name ensures AS will always
  be able to detect any such modification and rejects the message from TAP (by
  outputting $\bot$).  This is done by having TS include the circuit id $j$ in
  the encrypted payload $ct$ --- that TAP cannot modify. AS verifies that
  value against the circuit id forwarded by TAP, and any mismatch results in
  execution termination.} 
%
%
\change{
  Though TAP cannot tamper with $j$ without being detected, it could learn the popularity of certain rules by observing circuit id values (which are passed to TAP in plaintext to help find corresponding garbled circuit $F$ to execute).  We acknowledge that metadata attacks are a limitation in \name and we discuss a cover traffic approach to address them (\secref{sec:discussion}).
}

\paragraph{\em Security Analysis of TS and AS.}\label{sec:sec-ts}
We assume \TS and \AS are honest but curious. We define security as follows.

\begin{theorem}[$\privts$] \TS does not learn anything about the user constants $(\cone,\ctwo)$. \label{thm:TS} \end{theorem}
\noindent\textbf{Proof Sketch.}
\TS only receives from the client $\sects$ and $j$, which it uses to compute the
seed $\ei=(\eis,\eir)$. Thus, it can only learn the pairs of labels for all the
input wires (including the ones for user constants) to the garbled
circuit. \TS, \emph{by design}, does not have access to the client constants.\\[-5pt]


\AS should not learn about the user constants and the trigger data beyond what
is revealed from the output of the function~$f$.  Let $y_1=\fone(x,c)$ and
$y_2=\ftwo(x,c)$. We also need to ensure that when the output of the predicate
function $y_1=\false$, \AS does not learn the output of the
function $\ftwo$ and the payload $\dtp$.  We formally state these properties,
using the theorem below.
\begin{theorem}[$\privas^0$] If $y_1=\false$, then \AS learns nothing about $(x,c,\payload)$ other than what is revealed from $y_1=\false$. \label{thm:AS1} \end{theorem}
\begin{proof}
To know the value of $y_2$, \AS needs access to the decoding table
$d'$ (from the \textit{obliviousness} guarantee of garbled circuits
in~\cite{zahur2015two}).  \AS will be able to do this only if it has access to $L_1^{w_0}$ (from the IND-CCA security of the encryption
scheme). Note, $L_1^{w_0}$ is available to \TAP, and subsequently to \AS, only if
$\fone(\dt)=\true$ \cite{zahur2015two}. In case \TAP returns some garbage value other than $L_1^{w_0}$, the decryption still fails. Additionally, $\payload$ is protected by the IND-CCA security of the encryption
 scheme. 
\end{proof} 

\begin{theorem}[$\privas^1$] If $y_1=\true$, then for any ppt adversary $\advB$, the probability that
  $\advB$ wins the game $\privone$ is only negligibly more than random guessing. That is,
  \bnm \Prob{\privone = 1}\le 1/2 + \negl(\keylen)\,.\enm \label{thm:AS2}
\end{theorem}\vspace{-0.1in}
\begin{proof}
The indistinguishability of $ct^b$ follows from the semantic security of the encryption scheme.  Now note that $\privone$ is equivalent to $prv.sim_\mathcal{S})$ \cite{bellare2012foundations} in  \cite{zahur2015two}. The rest of the proof is based on the proof for the corresponding game ($prv.sim_\mathcal{S})$ in \cite{zahur2015two}). In fact in our setting, the view of the $\adv$ is a strict subset of that of the adversary presented in ~\cite{zahur2015two}. Specifically, our adversary $\adv$ does not have access to the garbled inputs $X^b,C^b$ and the garbled circuit $F$.  
Note that in the above game, a malicious \TAP instead of outputting $(Y, ct) \gets \textsf{TAPExec}\left((X, ct), (F, C)\right)$, could generate some arbitrary message. However, from the obliviousness property of garbled circuits (Thm. \ref{thm:TAP1}, we know that this message has to be  completely oblivious of $(F,X,c)$ and hence the privacy guarantee is upheld trivially. 
\end{proof}

\begin{proposition}[TAP Input Indistinguishability]
  For any ppt adversary $\advA$ with access to a circuit garbled with the scheme in \cite{zahur2015two}, the probability that $\advA$ distinguishes between a valid garbled input and randomly generated input is negligibly more than random guessing.
  \label{thm:input}
\end{proposition}
\noindent\textbf{Proof Sketch.}
Following Fig. 2 in \cite{zahur2015two}, it is clear that $\advA$ cannot validate inputs to XOR gates. For AND gates, the fact that at most one valid label for each input wire is revealed to $\advA$ and the correlated robustness of the hash function ensures that $F = (T_G, T_E)$ does not reveal information about the valid inputs. 

\subsection{Extracting and Replacing Substrings with Garbled Circuits} 
\label{app:regex}
\begin{figure*}[t]
\centering\gamesfontsize
\begin{tabular} {p{0.5em}p{7cm}p{5cm}rrrr} 
  \toprule
  \multirow{2}{*}{\#} & \multirow{2}{*}{\textbf{Rule description}} & \multirow{2}{*}{\textbf{Functions performed}} & \textbf{GC size} & \multicolumn{2}{c}{\textbf{Data transfer (KB)}} \\
  &  & & {\bf (KB)}& \textbf{\TS\send\TAP} & {\bf \TAP\send\AS} \\ 

  \midrule
  R1 & Share {your Tweets} (excluding replies) in Slack &  \texttt{! x[Text].startwith("@")}  & 0.2 & 43 & 20 \\[2pt] 
  R2 & Get Slack notifications for {new Twitter followers} with more than 5,000 followers &  \texttt{x[FollowerCount] > 5000}   & 1.0 & 29 & 3 \\[2pt] 
  R3 & Copy {New Events from Google Calendar} into iOS Calendar  & \texttt{x[StartTime] - x[EndTime]}   & 1.0 & 33 & 32   \\[2pt] 
  R4 & Blink your lights when you receive email from a specific address &  \texttt{x[Sender] == c}   & 5.8 & 29 & 3 \\[2pt] 
  R5 & Send SMS messages for {new Shopify orders}  & \texttt{x[Phone] != null}; \newline \texttt{x[Phone].replace(" ", "")}  & 9.0 & 27 & 3 \\[2pt] 
  R6 & Add {new inbound emails} as contacts in Ontraport & \texttt{x[SenderName].split(" ", 0)}; \newline \texttt{x[SenderName].split(" ", 1)}   & 30.5 & 34 & 13  \\[2pt] 
  R7 & Create Asana tasks when {new Slack messages} start with \texttt{\$request}  & \texttt{x[Text].startwith("\$request")}; \newline \texttt{x[Text].replace("\$request")}; \newline \texttt{c2.lookup(x[Channel])} & 92.4 & 29 & 4  \\[2pt] 
  R8 & Save {new liked Tweets} with links to Pocket &  \texttt{x[Text].contain("http")} & 173.4 & 43 & 20 \\[2pt] 
  R9 & Send SMS reminders for {upcoming Google Calendar events} & \texttt{x[Description].extract\_phone()}  & 4,668.9 & 51 & 28 \\[2pt] 
 R10 & Upload {new videos in Google Drive} to YouTube   &  \texttt{x[Filename].endwith("mp4|avi|mov")}  & 12.1 & 32,133 & 32,108 \\[2pt] 
  \bottomrule
\end{tabular}
\caption[]{Selected real-world rules for our experiments \change{from both IFTTT and Zapier}. We note the size of the
  corresponding garbled circuits, and the amount of data transferred from TS to TAP
  and TAP to AS during rule execution.}
\label{fig:rules-exp}
\end{figure*}

We now discuss how \name extends the regular expression matching technique described in~\secref{sec:boolfunc} to extract and replace substrings. 

\paragraph{Finding locations of matching substring.} Given a regular expression pattern \texttt{p}, the goal is to find the starting and ending positions of the matching substrings. 

Finding the ending positions can be achieved by applying the KMP algorithm~\cite{knuth1977fast} on the pattern $p$ to convert it into a DFA (denoted by $\dfa$), so that $\dfa$ will output an accepting state at the end of each matching substring. For example, if the pattern is \texttt{ab}, we will rewrite it as \texttt{.*ab} and convert the new pattern into DFA.
Then we use our matching protocol (\secref{sec:boolfunc}) to run $\dfa$ on the input string $\vecx$. However, instead of only checking whether the final state $\State_n$ is an accepting state, we check every state $\State_1, \dots, \State_n$ produced by $\dfa$. We denote the resulting $n$-bit sequence as $e_1, \dots, e_n$. If $e_i=1$, it indicates that the $i$-th bit is the end of a matching substring.

Since a DFA can only report the end positions of matches end, we need another
DFA to find the starting positions. 
We therefore compute a DFA $\Gamma'$ on the  
reversed pattern $p$. 
If we run $\Gamma'$ on the reversed  input
string, we get the 
beginning of the matching substring. Then, like the previous step, we
run $\Gamma'$ backward on $\vecx$ (by feeding from $x_n$ to $x_1$) and check
the type of every state to generate $b_n, \dots, b_1$. If $b_i=1$, it indicates
that the $i$-th bit is the beginning of a matching substring.

Finally, we can find the locations of all matching substrings. That is, we need to compute another $n$-bit sequence $m_1, \dots, m_n$ where $m_i = 1$ if and only if the $i$-th bit is part of a matching substring.

We can observe that $m_1 = b_1$ and for any $i$ such that $2 \le i \le n$, $m_i$ can be calculated as 
$ m_i = b_i \vee (\neg e_{i-1} \wedge m_{i-1}). $


\paragraph{Extracting matching substring.} 
To extract the matching substrings, we want to replace the characters in non-matching parts with the padding character (\texttt{0x00}).
Therefore,  the output string $\vec{y} = \{ y_1, \dots, y_n \}$ is computed by $y_i= m_i \wedge x_i$.

\newcommand{\ttt}{\texttt{t}}
\paragraph{Replacing matching substring.}
In our dataset, all \texttt{replace(s,t)} functions are used with \texttt{t} set to empty string, so it is equivalent to removing the matching substring, and thus the output string $\vec{y} = \{ y_1, \dots, y_n \}$ is computed by $y_i= \neg m_i \wedge x_i$.

However, for completeness, we will describe a protocol for the general case scenario where $|\ttt|>0$, where $\ttt$ denotes the size of the string $\ttt$. 
The output string size will be 
$n\times {|\texttt{s}|\over|\ttt|}$
since the TAP should not know which substring is matched and replaced and should assume all substrings can be replaced. When 
$|\texttt{s}| \gg |\ttt|$
the sizes of the resulting garbled circuits will be unbearably large. Therefore, we purpose an alternative design approach where the actual replacement is processed in the action service: we replace the first character of each matching substring with some placeholder character, say \texttt{0xff}, and the rest with the padding character \texttt{0x00}, so the action service can invoke the following functions to complete the replacement:
\texttt{y.replace("0x00", ""); y.replace("0xff", t);} 
where \texttt{y} is the decoded output string. Note the first \texttt{replace()} is required regardless of our protocol, since it is needed for removing the padding from the input string. 

We argue this approach does not break our security goal, revealing no
additional trigger data that is not supposed to be revealed to the action
service. If the replacement string $t$ is considered sensitive 
the client can
encrypt the replacement mapping with the 1 label of the output bit corresponding
to $\bigvee_{i=1}^n m^i$, similar to how we protect $d$ and $k$
in~\figref{fig:rules-exp}.

Assuming an ASCII encoding and \texttt{0xff} as the placeholder character, we can compute the output string $\vec{y}$ using $ y_i=s_{i-(i-1 \mod 8)} \vee (\neg m_i \wedge x_i),$
where the $i-(i-1 \mod 8)$-th bit is the first bit of the character that $i$-th bit belongs.

\subsection{Implementing Supported Function}
\label{app:func-impl}

In this appendix section, we describe how to implement each operation that appears in~\figref{fig:functions}, except for Boolean and arithmetic operations, since existing GC frameworks like EMP toolkit~\cite{emp-toolkit} already provide built-in functions to efficiently translate them.

\noindent{} \texttt{x == s} and \texttt{x.startwith(s)}. A bit-wise comparison between \texttt{x} and \texttt{s} is performed up to the \texttt{min(len(x), len(s))} bit, and results are feed into a large AND gate as output. For \texttt{x == s}, We additionally check if the next remaining character in \texttt{x} or \texttt{s} is a padding character.

\noindent{} \texttt{x.endwith(s)} and \texttt{x.contain(s)}. These functions need to be first converted to a correspondingly regular expression and then matched against \texttt{x}.


\noindent \texttt{x.replace(s, t)}. We can apply the DFA replacement technique described in Appendix \ref{app:regex} directly for this type of functions.

\noindent\texttt{x.extract\_phone()} and \texttt{extract\_email()}. We apply the DFA extraction described in Appendix \ref{app:regex} by constructing appropriate regular expressions.
However, as we need the matching results to be non-overlapping, one modification is needed : we can append \texttt{[\string^a-Z0-9]}
to the regular expression and shift the final matching position forward by 1 character.

\noindent{} \texttt{x.split(d,i)}. Without loss of generality, we assume \texttt{d} is a single character. First we need to create two regular expressions, $\dfa_1$ and $\dfa_2$, to that output accepting states when the \texttt{i}-th and \texttt{i+1}-th occurrences of \texttt{d} is encountered. Once we have the starting and ending location of the substring, we can proceed with substring extractions.

\noindent{} \texttt{x.truncate(n)}. We can keep a variable counter \texttt{c} that gets increased after each bit in \texttt{x} is processed. And each output bit \texttt{y[i]} is computed by \texttt{x[i] \& (n > c)}.

\noindent{} \texttt{x.tolowercase()}. Assume ASCII encoding, for each character in \texttt{x}, we first check if the last five bits are in the valid ranges; if so, we flip the sixth bit.

\noindent{} \texttt{m.lookup(x)}. First, we compare \texttt{x} with each key of \texttt{m}, and store the matching results into a \texttt{len(m)}-bit sequence. We denote this sequence as \texttt{b}. Then, the output \texttt{y} is computed iteratively by \texttt{y = (b[i] \& v[i]) | (!b[i] \& y)} as \texttt{i} ranges from 1 to \texttt{len(m)}.

\subsection{Rules used for performance evaluation}\label{sec:perf-appendix}
The descriptions for the selected trigger-action rules we evaluated in Section \ref{sec:macro} are shown in \figref{fig:rules-exp}.



\end{document}